\pdfoutput=1
\documentclass{amsart}

\usepackage{amscd}
\usepackage{bbm}
\usepackage{mathrsfs}
\usepackage{eucal}
\usepackage{faktor}
\usepackage{xfrac}
\usepackage{braket}
\usepackage{relsize} 
\usepackage{nccmath}
\usepackage{mathtools}
\usepackage{comment}
\usepackage{verbatim}
\usepackage{wasysym}
\usepackage{graphicx}
\usepackage{caption}
\usepackage{import}
\usepackage{latexsym}
\usepackage{amsfonts}
\usepackage{amssymb}
\usepackage{amsmath}
\usepackage{mathtext}
\usepackage{cite}
\usepackage{enumerate}
\usepackage{float}
\usepackage{amsthm}
\usepackage{upgreek}
\usepackage{nicefrac}
\usepackage{multicol}
\usepackage{rotating}
\usepackage{array}
\usepackage{microtype}
\usepackage{vwcol} 

\usepackage{etoolbox}
\makeatletter
\pretocmd{\chapter}{\addtocontents{toc}{\protect\addvspace{15\p@}}}{}{}
\pretocmd{\section}{\addtocontents{toc}{\protect\addvspace{5\p@}}}{}{}
\makeatother
\let\oldtocsection=\tocsection
\let\oldtocsubsection=\tocsubsection
\let\oldtocsubsubsection=\tocsubsubsection
\renewcommand{\tocsection}[2]{\hspace{0em}\oldtocsection{#1}{#2}}
\renewcommand{\tocsubsection}[2]{\hspace{1.8em}\oldtocsubsection{#1}{#2}}
\renewcommand{\tocsubsubsection}[2]{\hspace{4.4em}\oldtocsubsubsection{#1}{#2}}

\usepackage[svgnames]{xcolor}
\definecolor{linkcolor}{HTML}{e88d67} 
\definecolor{citecolor}{HTML}{e88d67} 
\definecolor{urlcolor}{HTML}{e88d67} 
\definecolor{myNewColorA}{HTML}{fec3a6}
\definecolor{myNewColorB}{HTML}{ffaf80}
\definecolor{myNewColorC}{HTML}{fb8f67}
\definecolor{seagreen}{HTML}{337180}
\definecolor{mseagreen}{HTML}{369673}
\definecolor{darksalmon}{HTML}{e88d67}
\definecolor{silver}{HTML}{bbbbbb}
\definecolor{flowerblue}{HTML}{4e77fc}
\definecolor{tomato}{HTML}{ff6347}
\definecolor{orange}{HTML}{f2b13d}
\definecolor{darkgray}{HTML}{939393}

\usepackage[
  bookmarks = true,
  bookmarksopen = true,
  colorlinks = , 
  urlcolor = urlcolor, 
  linkcolor = linkcolor, 
  citecolor = citecolor, 
  linktoc = all,
  pagebackref = true
]{hyperref} 
\mathtoolsset{showonlyrefs,showmanualtags}

\DeclareMathOperator{\Mat}{Mat}

\newcommand{\brackets}[1]{\left( #1 \right)}
\newcommand{\PoissonBrackets}[1]{\left\{ #1 \right\}}
\newcommand{\LieBrackets}[1]{\left[ #1 \right]}
\newcommand{\angleBrackets}[1]{\langle #1 \rangle}

\newcommand{\quot}[2]{{\raisebox{.3em}{$#1\!$}\bigl/\raisebox{-.3em}{$\!#2$}}}
\newcommand{\Painleve}{Painlev{\'e} }

\theoremstyle{plain}
\newtheorem{thm}{Theorem}[]

\theoremstyle{definition}

\theoremstyle{remark}
\newtheorem{rem}{Remark}

\usepackage{chngcntr}
\counterwithin*{thm}{section} 
\counterwithin*{lem}{section} 
\counterwithin*{claim}{section} 
\counterwithin*{prop}{section} 
\counterwithin*{corl}{section} 
\counterwithin*{defn}{section} 
\counterwithin*{exmp}{section} 
\counterwithin*{ex}{section} 
\counterwithin*{rem}{section} 

\usepackage{apptools}
\AtAppendix{\counterwithin{thm}{section}}
\AtAppendix{\counterwithin{lem}{section}}
\AtAppendix{\counterwithin{claim}{section}}
\AtAppendix{\counterwithin{prop}{section}}
\AtAppendix{\counterwithin{corl}{section}}
\AtAppendix{\counterwithin{defn}{section}}
\AtAppendix{\counterwithin{exmp}{section}}
\AtAppendix{\counterwithin{ex}{section}}
\AtAppendix{\counterwithin{rem}{section}}

\allowdisplaybreaks

\usepackage{tikz}

\begin{document}

    \title[]{Affine Weyl groups and non-Abelian discrete systems: 
    \\
    an application to the $d$-Painlevé equations}

    \author{Irina Bobrova}
    \noindent\address{\noindent 
    Max-Planck-Institute for Mathematics in the Sciences, 04103, Leipzig, Germany
    }
    \email{ia.bobrova94@gmail.com}

    
    \maketitle
    
    \begin{abstract}
    A non-abelian generalisation of a birational representation of affine Weyl groups and their application to the discrete dynamical systems is presented. By using this generalisation, non-commutative analogs for the discrete systems of $A_n^{(1)}$, $n \geq 2$ type and of $d$-Painlevé equations with an additive dynamic were derived. A coalescence cascade of the later is also discussed. 
    \end{abstract}
    
    \tableofcontents
    
    \section{Introduction}

In the series of papers by K. Okamoto \cite{okam1}, \cite{okam2}, \cite{okam3}, \cite{okam4}, it was shown that the symmetries of the differential Painlevé equations form a group structure isomorphic to a certain affine Weyl group. These symmetries preserve a class of the equation, but change the parameters. Thanks to the affine type, one can introduce a translation operator that leads to a discrete dynamic of the additive type. Moreover, H. Sakai had showed that these symmetries are closely connected with rational surfaces and, as a result, classified discrete Painlevé equations according to a rational surface type. His classification contains 22 discrete Painlevé equations. These discrete systems are either of elliptic, multiplicative or additive type and are known as $ell$-, $q$- or $d$-Painlevé equations respectively. In fact, there is another approach to the discrete Painlevé equations that does not involve the geometric methods. It is based on the afﬁne Weyl groups and is developed by M. Noumi and Y. Yamada \cite{noumi1998affine}. The authors have defined discrete systems by using a birational representation of the afﬁne Weyl groups. Namely, a proper extension of this representation to the ﬁeld of rational functions defines dynamical variables and, thus, a discrete system, thanks to the translation operators. This approach has a straightforward generalisation to the non-commutative case, while the geometric method is needed to be developed. Here we present a non-commutative generalisation of the affine Weyl groups' application to the discrete systems and support it by a series of examples. Examples given in Section \ref{sec:Ansys} are non-commutative analogs of the discrete systems of $A_n^{(1)}$, $n \geq 2$ type (see \cite{noumi1998affine}, \cite{noumi2000affine}). As in the commutative case, generators of the corresponding birational representation are Bäcklund transformations of the dressing chain in the Noumi-Yamada variables also known as a higher order analog of the fourth and fifth Painlevé equations. Section \ref{sec:dP} collects examples, which might be regarded as non-abelian versions for $d$-Painlevé equations. We leave their study in the geometric framework developed by H. Sakai for the forthcoming papers. 

Let us briefly discuss the main idea of the paper \cite{noumi1998affine}. The Painlevé equations, differential and difference, have symmetries that form an affine Weyl group. Generators $s_i$, $i \in I := \{ 0, 1, \dots, n\}$ of this group act on the corresponding root lattice $Q = \mathbb{Z}\PoissonBrackets{\alpha_0, \dots, \alpha_n}$ by automorphisms of the field $\mathbb{C} (\alpha)$ of rational functions in $\alpha_i$, $i \in I$. Let us consider a set of ``variables'' $f_i$, $i \in I$ and an extension of the field $\mathbb{C}(\alpha)$ to the field $\mathbb{C} (\alpha, f) = \mathbb{C}(\alpha) (f_i, \, i \in I)$ of rational functions in $\alpha_i$ and $f_i$. We assume that a specification of the $s_i$-action on $f_j$ preserve the Weyl group structure for any $i, j \in I$. Then, defining a translation element $t_{\mu} \in W$ and a set of rational functions $F_{\mu, i} (\alpha, f) \in \mathbb{C} (\alpha, f)$, one can obtain a discrete dynamical system, where $\alpha_i$ are discrete time variables and $f_i$ are depended variables. Therefore, a suitable birational representation of the affine Weyl group leads to a discrete dynamical system. Note that classes of such representations arise naturally from Bäcklund transformations of ordinary differential equations, in particular, of the Painlevé equations.

The differential Painlevé equations posses different non-abelian analogs. The matrix Hamiltonian analogs had been derived by H. Kawakami \cite{Kawakami_2015}. Then, his list was extended by several papers. Some of them contain examples of analogs with non-commutative parameters \cite{Adler_Sokolov_2020_1}, \cite{Bobrova_2022}, \cite{bobrova2022classification}, while a full classification of systems with abelian parameters is done in \cite{bobrova2023classification}. In fact, some non-abelian analogs of the difference Painlevé equations have already appeared (see, e.g., \cite{cassatella2014singularity}, \cite{adler2020}). In \cite{bobrova2022fully}, it had been shown that an analog of the fourth Painlevé equation admits Bäcklund transformations forming an affine Weyl group structure similar to the commutative case. Later, in \cite{bershtein2023hamiltonian}, the authors derived Bäcklund transformations for matrix Hamiltonian Painlevé systems of all types and showed that they form affine Weyl groups as in the classical case. Our study of discrete non-abelian $d$-Painlevé equations has been inspired by the latter paper. Now let us present non-commutative setting and describe main results of the current article. 

\medskip
Since we are working on non-abelian analogs of the ordinary differential and difference equations, we define an associative unital division ring $\mathcal{R}$ over the field $\mathbb{C}$ equipped with a derivation. We assume that all greek letters belong to the field $\mathbb{C}$, while the elements $f_i$ are from $\mathcal{R}$. We will often call $f_i$ as \textit{functions}. The derivation $d_t : \mathcal{R} \to \mathcal{R}$ of the ring $\mathcal{R}$ is a $\mathbb{C}$-linear map satisfying the Leibniz rule. We also assume that there is a central element $t$ such that $d_t(t) = 1$ and for any $\alpha \in \mathbb{C}$ we have $d_t(\alpha) = 0$. Here and below we identify the unit of the field with the unit of the ring. For the brevity we denote $d_t(f_i) = \dot{f}_i$, $d_t^2(f_i) = \Ddot{f}_i$, and so on. Note that on $\mathcal{R}$ we have an involution called the transposition $\tau$, which acts trivially on the generators of $\mathcal{R}$ and for any elements $F$, $G \in \mathcal{R}$ we have $\tau(F \, G) = \tau(G) \, \tau(F)$. This involution can be naturally extended to the matrices over $\mathcal{R}$. We would rather not specify the generators of the ring $\mathcal{R}$ in order to avoid overloaded description of a pretty simple thing. Instead, we encourage the reader to think of the ring $\mathcal{R}$ as a generalization of rational functions over the field $\mathbb{C}$ to a non-abelian case. 

Now it is clear how to generalise the construction from \cite{noumi1998affine} to the non-abelian case. 
We just consider $f_i \in \mathcal{R}$ and repeat the same arguments as above. More accurate and explicit description can be found in Section~\ref{sec:affw_dsys}. It gives us a non-abelian version of the discrete dynamical systems, an application of which we present in Sections \ref{sec:Ansys} and \ref{sec:dP}. 

Section \ref{sec:Ansys} is devoted to the introduction and study of non-abelian discrete systems of type $A_n^{(1)}$, $n \geq 2$, whose commutative versions have been discussed in the papers \cite{noumi1998affine}, \cite{noumi2000affine}. According to Section \ref{sec:affw_dsys}, we define a birational representation of the extended affine Weyl group of type $A_n^{(1)}$ (see Theorem \ref{thm:Aldsys}), the corresponding translation operators of which define discrete dynamics of $d$-type. Note that this birational representation is closely connected with the dressing chain \cite{veselov1993}. Namely, the dressing chain might be rewritten in the Noumi-Yamada variables \cite{takasaki2003spectral} and, thus, the birational representation gives us Bäcklund transformations of the system. Moreover, the systems of $A_n^{(1)}$, $n \geq 2$ type are higher order generalisations for the fourth and fifth Painlevé equations. We have derived a non-commutative analog of the dressing chain in the Noumi-Yamada variables (see Theorem \ref{thm:Ansys}) and present its Lax pair (see Theorem \ref{thm:Anpair}). Our results are generalisations of the quantum $A_n^{(1)}$, $n \geq 2$ type systems studied in \cite{nagoya2004quantum}. Unlike the paper \cite{nagoya2004quantum}, we do not assume any algebraic relations for the variables and present an explicit form of the Lax pair. 

Further natural examples are related to the Painlevé equations and are discussed in Section \ref{sec:dP}. As we have already mentioned, Bäcklund transformations of the differential Painlevé equations form an affine Weyl group and, as a result, define discrete dynamical systems, since a specification of $s_i$-actions appears naturally from the symmetries. In the paper \cite{bershtein2023hamiltonian}, the authors had presented Bäcklund transformations for non-abelian analogs of the differential Painlevé equations. By using them, we obtained non-commutative analogs of all $d$-Painlevé equations listed in Section 7 from the paper \cite{sakai2001rational} (see also Section 8 in \cite{kajiwara2017geometric}). 
\begin{minipage}{.4\textwidth}
\vspace{0.5mm}
Our list of non-abelian $d$-Painlevé equations is given in Appendix \ref{app:dP}, while their detailed discussions are in Section~\ref{sec:dP}. These equations are connected with each other by a degeneration procedure (see degeneration data in Appendix \ref{app:deg_data}). 
A coalescence cascade coincides with the well-known one (see \cite{sakai2001rational}).
\end{minipage}
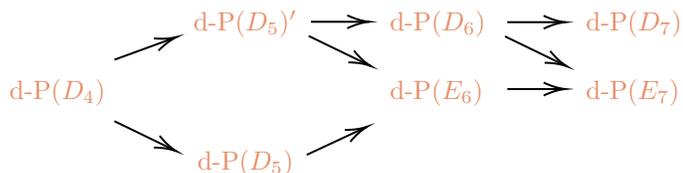
\begin{minipage}{0.6\textwidth}
\vspace{-3mm}
\begin{figure}[H]
    \centering
    \scalebox{1.}{\tikzset{every picture/.style={line width=0.75pt}} 

\begin{tikzpicture}[x=0.75pt,y=0.75pt,yscale=-1,xscale=1]

\draw    (290.63,69.79) -- (320.13,59.01) ;
\draw [shift={(322.01,58.33)}, rotate = 159.93] [color={rgb, 255:red, 0; green, 0; blue, 0 }  ][line width=0.75]    (10.93,-3.29) .. controls (6.95,-1.4) and (3.31,-0.3) .. (0,0) .. controls (3.31,0.3) and (6.95,1.4) .. (10.93,3.29)   ;
\draw    (387.63,117.79) -- (416.34,104.4) ;
\draw [shift={(418.15,103.56)}, rotate = 155] [color={rgb, 255:red, 0; green, 0; blue, 0 }  ][line width=0.75]    (10.93,-3.29) .. controls (6.95,-1.4) and (3.31,-0.3) .. (0,0) .. controls (3.31,0.3) and (6.95,1.4) .. (10.93,3.29)   ;
\draw    (290.63,100.79) -- (320.05,115.05) ;
\draw [shift={(321.85,115.92)}, rotate = 205.86] [color={rgb, 255:red, 0; green, 0; blue, 0 }  ][line width=0.75]    (10.93,-3.29) .. controls (6.95,-1.4) and (3.31,-0.3) .. (0,0) .. controls (3.31,0.3) and (6.95,1.4) .. (10.93,3.29)   ;
\draw    (388.63,57.79) -- (418.05,72.05) ;
\draw [shift={(419.85,72.92)}, rotate = 205.86] [color={rgb, 255:red, 0; green, 0; blue, 0 }  ][line width=0.75]    (10.93,-3.29) .. controls (6.95,-1.4) and (3.31,-0.3) .. (0,0) .. controls (3.31,0.3) and (6.95,1.4) .. (10.93,3.29)   ;
\draw    (389.78,50.63) -- (417.01,50.63) ;
\draw [shift={(419.01,50.63)}, rotate = 180] [color={rgb, 255:red, 0; green, 0; blue, 0 }  ][line width=0.75]    (10.93,-3.29) .. controls (6.95,-1.4) and (3.31,-0.3) .. (0,0) .. controls (3.31,0.3) and (6.95,1.4) .. (10.93,3.29)   ;
\draw    (487.63,57.99) -- (517.05,72.25) ;
\draw [shift={(518.85,73.12)}, rotate = 205.86] [color={rgb, 255:red, 0; green, 0; blue, 0 }  ][line width=0.75]    (10.93,-3.29) .. controls (6.95,-1.4) and (3.31,-0.3) .. (0,0) .. controls (3.31,0.3) and (6.95,1.4) .. (10.93,3.29)   ;
\draw    (488.78,50.83) -- (516.01,50.83) ;
\draw [shift={(518.01,50.83)}, rotate = 180] [color={rgb, 255:red, 0; green, 0; blue, 0 }  ][line width=0.75]    (10.93,-3.29) .. controls (6.95,-1.4) and (3.31,-0.3) .. (0,0) .. controls (3.31,0.3) and (6.95,1.4) .. (10.93,3.29)   ;
\draw    (488.8,84.58) -- (516.03,84.58) ;
\draw [shift={(518.03,84.58)}, rotate = 180] [color={rgb, 255:red, 0; green, 0; blue, 0 }  ][line width=0.75]    (10.93,-3.29) .. controls (6.95,-1.4) and (3.31,-0.3) .. (0,0) .. controls (3.31,0.3) and (6.95,1.4) .. (10.93,3.29)   ;

\draw (329.13,43.37) node [anchor=north west][inner sep=0.75pt]    {\ref{eq:dPD5'}};
\draw (331.1,113.73) node [anchor=north west][inner sep=0.75pt]    {\ref{eq:dPD5}};
\draw (428.1,78.13) node [anchor=north west][inner sep=0.75pt]    {\ref{eq:dPE6}};
\draw (428.1,43.4) node [anchor=north west][inner sep=0.75pt]    {\ref{eq:dPD6}};
\draw (236.35,78.12) node [anchor=north west][inner sep=0.75pt]    {\ref{eq:dPD4}};
\draw (527.1,78.33) node [anchor=north west][inner sep=0.75pt]    {\ref{eq:dPE7}};
\draw (527.1,43.6) node [anchor=north west][inner sep=0.75pt]    {\ref{eq:dPD7}};

\end{tikzpicture}}
\end{figure}
\end{minipage}
\\[0.5mm]
Moreover, these equations are expected to be integrable in the sense of the Lax pairs and to be Hamiltonian (see Remarks \ref{rem:lax} and \ref{rem:ham}). We leave these problems and a study of non-abelian analogs of the $q$-Painlevé equations for the further research.

\begin{rem}
\label{rem:lax}
Regarding the \ref{eq:dPE7} system with the central element $t$, 
\begin{gather}
    \tag*{\ref{eq:dPE7}}
    \begin{gathered}
    \begin{aligned}
        \bar \alpha_0 
        = \alpha_0 - 1,&
        &&&&&
        &\bar \alpha_1 
        = \alpha_1 + 1,
        \\
        \bar q
        + q 
        = - \alpha_1 p^{-1},&
        &&&&&
        &\bar p
        + p 
        = t + 2 {\bar{q}}^2,
    \end{aligned}
    \end{gathered}
\end{gather}
its discrete isomonodromic Lax pair is given by the matrices $\mathcal{A}_n = \mathcal{A}_n (\lambda)$, $\mathcal{B}_n = \mathcal{B}_n (\lambda)$ depending on the commutative spectral parameter $\lambda$, i.e. $\lambda \in \mathcal{Z} (\mathcal{R})$:
\begin{align}
    \mathcal{A}_n
    &= 
    \begin{pmatrix}
        1 & 0 \\[0.9mm]
        0 & -1
    \end{pmatrix}
    \lambda^2
    + 
    \begin{pmatrix}
        0 & 1 \\[0.9mm]
        2 p & 0
    \end{pmatrix}
    \lambda
    + 
    \begin{pmatrix}
        - p + \tfrac12 t & - q \\[0.9mm]
        2 p q + 2 \alpha_1 & p - \tfrac12 t
    \end{pmatrix}
    ,
    &
    \mathcal{B}_n
    &= 
    \begin{pmatrix}
        - 2 & 0 \\[0.9mm] 0 & 0
    \end{pmatrix}
    \lambda
    + 
    \begin{pmatrix}
        - 2 q & - 1 \\[0.9mm] - 2 \bar p & 0
    \end{pmatrix}
\end{align}
that are a non-commutative generalisation of those presented in \cite{grammaticos1998degeneration} (see Appendix A.9 therein). 
One~can verify that the compatibility condition of the discrete linear problem for the function $Y_n = Y_n (\lambda)$
\begin{align}
    &&
    &\left\{
    \begin{array}{lcl}
         \partial_{\lambda} Y_{n}
         &= \mathcal{A}_n \, Y_n,  
         \\[3mm]
         Y_{n + 1}
         &= \mathcal{B}_n \, Y_n
    \end{array}
    \right.
    &
    \Rightarrow&
    &
    \partial_{\lambda} \mathcal{B}_{n}
    = \mathcal{A}_{n + 1} \, \mathcal{B}_n
    - \mathcal{B}_n \, \mathcal{A}_n
    &&
\end{align}
is equivalent to the \ref{eq:dPE7}, since the commutator $[p, q]$ is invariant under the map
\begin{align}
    &&
    \psi
    &: \mathcal{R}^2 \to \mathcal{R}^2,
    &
    (q, p)
    &\mapsto \brackets{\bar q, \bar p}
    = \brackets{
    - p + t + 2 q^2, \,
    - q - \bar \alpha_1 (- p + t + 2 q^2)^{-1}
    }
    .
    &&
\end{align}
Once $t \in \mathcal{R}$, the commutator $[p, q]$ is no longer a conserved quantity. The latter fact might have been caused by the Hamiltonian structure similarly to the non-abelian Hamiltonian ODEs (see Lemma 1 in \cite{bobrova2022classification} and its generalisation, Lemma 2.1, in \cite{bobrova2023equations}). 
\end{rem}

\begin{rem}
\label{rem:ham}
Derived systems from Appendix \ref{app:dP} are expected to be Hamiltonian. Following the discrete Hamiltonian setting described in the paper \cite{mase2020discrete}, one can introduce its non-commutative analog as follows. Let us consider the non-commutative partial derivatives as in Section 2.1.2 from \cite{bobrova2023equations}, which allow us to work on non-autonomous Hamiltonian systems. Then, a discrete system in $q$, $p$ variables is Hamiltonian if there exists a function $H = H (q, \bar p)$ such that the system can be rewritten in the form
\begin{align}
    \label{eq:dhamsys}
    &&
    p
    &= \partial_q H,
    &
    \bar q
    &= \partial_{\bar p} H.
    &&
\end{align}
For the \ref{eq:dPE7} system, a Hamiltonian is $H = - q \, \bar p + t q + \tfrac13 q^3 - \bar \alpha_1 \, \ln \bar p$, where for the symbol $\ln f$ we define the right logarithmic derivative by $d_t (\ln f) := f^{-1} \, \dot{f}$\footnote{Its left analog can be easily derived by using the $\tau$-action.}. Then, the non-abelian derivatives are
\begin{align}
    &&
    \partial_q H
    &= - \bar p + q^2 + t,
    &
    \partial_{\bar p} H
    &= - q - \bar \alpha_1 \, \bar p^{-1},
    &
    \partial_t H
    &= q
    &&
\end{align}
and \eqref{eq:dhamsys} is equivalent to the \ref{eq:dPE7} system.
\end{rem}

\subsection*{Acknowledgements}
The author is deeply grateful to N. Safonkin for useful remarks. 

\section{Affine Weyl groups and discrete dynamics}
\label{sec:affw_dsys}

Here we will briefly discuss necessary info related to the affine Weyl groups and discrete dynamics generated by the translations. We will follow the paper \cite{noumi1998affine}, the main idea of which is that an extended birational representation of the affine Weyl group leads to a discrete system. We will transfer this idea to the non-abelian setting. Some foundations of the commutative theory can be found in \cite{kac1990infinite} (see also \cite{shi2022translations}), while its application to the Painlevé equations might be found in \cite{joshi2019discrete}. Non-abelian setting is already discussed in the Introduction. 

\begin{rem}
We would like to mention the paper \cite{berenstein2008lie}, that is close to our subject and where the authors had introduced and studied Lie algebras and Lie groups over non-commutative rings.
\end{rem}

\medskip
Let us fix a generalized Cartan matrix $C = (c_{ij})$, where $i, j \in I := \PoissonBrackets{0, 1, \dots, n}$, i.e. it is of affine type. 
Sets $\Delta = \{\alpha_0, \dots, \alpha_n \}$, $\Delta^\vee = \{\alpha_0^\vee, \dots, \alpha_n^\vee \}$ correspond to simple roots and simple co-roots associated with the matrix $C$, where $\alpha_0$ and $\alpha_0^\vee$ are simple affine root and co-root respectively. The sets $\Delta$ and $\Delta^\vee$ form a basis for the vector spaces $V$ and $V^*$ respectively. Denote by $Q = Q(C)$ and $Q^\vee = Q^\vee(C)$ the root and co-root lattices
\begin{align}
    &&
    Q
    &= \mathbb{Z} \, \Delta,
    &
    Q^\vee
    &= \mathbb{Z} \, \Delta^\vee.
    &&
\end{align}
Recall that the pairing $\angleBrackets{\, \cdot \, , \cdot \,}: Q \times Q^\vee \to \mathbb{Z}$ is defined by $\angleBrackets{\alpha_i, \alpha_j^\vee} = c_{ij}$ and $\alpha_i^\vee = 2 \alpha_i / (\alpha_i, \alpha_i)$. 

Denote by $W = W(C)$ the Weyl group (or the Coxeter group) defined by generators $s_i$, $i \in I$:
\begin{align}
    W(C)
    &= \angleBrackets{
    s_0, s_1, \dots, s_n \, 
    \big| \, 
    s_i^2
    = 1, \,
    (s_i \, s_j)^{m_{ij}}
    = 1
    },
\end{align}
where the exponents are determined by the value of the product $c_{ij} c_{ji}$ as below
\vspace{-2mm}
\begin{table}[H]
    \centering
    \begin{tabular}{c||ccccc}
         $c_{ij} c_{ji}$
         & 0
         & 1
         & 2
         & 3
         & $\geq 4$
         \\
         \hline\hline
         $m_{ij}$
         & 2 
         & 3
         & 4
         & 6
         & $\infty$
    \end{tabular}
\end{table}
\vspace{-2mm}
These generators act naturally on $Q$ by reflections
\begin{align}
    \label{eq:refl}
    s_i (\alpha_j)
    &= \alpha_j - \angleBrackets{\alpha_i, \alpha_j^\vee} \, \alpha_i
    = \alpha_j - c_{ij} \, \alpha_i.
\end{align}

Note that each $s_i$-action on the lattice $Q$ induces an automorphism of the field $\mathbb{C} (\alpha) = \mathbb{C} (\alpha_i, \, i \in I)$ of rational functions in $\alpha_i$. Hence, $\mathbb{C}(\alpha)$ is a left $W$-module. 

Now we are going to introduce new set of ``variables'' and, as a result, define an extension of the field $\mathbb{C}(\alpha)$. Let us consider the set of elements $f_i \in \mathcal{R}$, $i \in I$, which we will often call either \textit{functions} or \textit{variables}. We propose an extension of the representation of $W$ on $\mathbb{C}(\alpha)$ to the field $\mathbb{C} (\alpha, f) = \mathbb{C}(\alpha) (f_i, \, i \in I)$ of rational functions in $\alpha_i$ and $f_i$, $i \in I$. One needs to specify the action of $s_i$ on $f_j$ in such a way that the automorphisms $s_i$ on $\mathbb{C} (\alpha, f)$ are involutions and satisfy the braid relations, i.e. they must preserve the Weyl group structure. 

\begin{rem}
\label{rem:BT_dP}
Such classes of representations arise naturally from Bäcklund transformations of non-abelian analogs of the differential Painlevé equations. We will consider them in details in Section \ref{sec:dP}. 
\end{rem}

\begin{rem}
Sometimes it is necessary to work with an extended Weyl group $\widetilde W$. According to the definition, it is a semi-direct product of the Weyl group and the group $\Omega$ of automorphisms of the Dynkin diagram $\Gamma(C)$, i.e. $\widetilde W = W \rtimes \Omega$. Recall that an automorphism of $\Gamma(C)$ is a bijection $\pi$ on $I$ such that $c_{\pi(i) \pi(j)} = c_{ij}$. Hence, the commutative relations are given by $\pi \, s_i = s_{\pi(i)} \, \pi$. The representations of $W$ lifts to a representation of~$\widetilde{W}$.
\end{rem}

Recall that one of the important property of the affine Weyl groups is that they have translations, also known as \textit{Kac translations}. Let $W_0$ be a finite Weyl group, $\delta = \sum_{i \in I} k_i \, \alpha_i$ be the \textit{null root} and $V_0 = \PoissonBrackets{\mu \in V \, \big| \, \angleBrackets{\mu, \delta^\vee} = 0}$. For an element $\mu \in V_0$ such that $\angleBrackets{\mu, \mu^\vee} \neq 0$ we define a \textit{translation element} $t_{\mu} \in W$ by the formula
\begin{align}
    \label{eq:tr_el}
    t_{\mu}
    &= s_{\delta - \mu} \, s_{\mu}
\end{align}
and suppose that $w \, t_{\mu} = t_{w(\mu)} \, w$ for any $w \in W$. Note that, in particular, $t_{\alpha} t_{\beta} = t_{\alpha + \beta}$ and then $t_{\alpha} t_{\beta} = t_{\beta} t_{\alpha}$. This operator acts on simple affine roots as follows
\begin{align}
    \label{eq:shift}
    t_{\mu} (\alpha)
    &= \alpha
    - \angleBrackets{\mu, \alpha} \, \delta
    = \alpha - \mu_{\alpha} \delta
    .
\end{align}
It is known that the affine Weyl group is decomposed into a semi-direct product of translations in the \textit{lattice part} $M$ and the finite Weyl group $W_0$ acting on $M$, i.e. $W = M \rtimes W_0$. The lattice part $M$ acts on $\mathbb{C}(\alpha)$ as a shift operator, thanks to \eqref{eq:shift}. Since the null root $\delta$ is $W$-invariant, it is convenient to set it to be a nonzero constant. Therefore, it represents the scaling of the lattice $M$. 

Turning to the discrete systems, suppose that we extended the action of $W$ from $\mathbb{C}(\alpha)$ to $\mathbb{C}(\alpha, f)$. Here we consider an arbitrary extension $\mathbb{C}(\alpha, f)$ as a left $W$-module, assuming that each element of $W$ acts on the function ﬁeld as an automorphism. For each $\mu \in M$ we define a set of rational functions $F_{\mu, i} (\alpha, f) \in \mathbb{C}(\alpha, f)$~by
\begin{align}
    \label{eq:dys}
    t_{\mu} (f_i)
    &= F_{\mu, i} (\alpha, f).
\end{align}
This set can be considered as a \textit{discrete dynamical system}. Here $\alpha_i$ and $f_i$ are discrete time variables and the depended variables respectively. Based on the action of $t_\mu$ on the discrete time variables, discrete dynamics can be classified into additive ($d$-equations), multiplicative ($q$-equations), or elliptic ($ell$-equations) types. We~remark that a similar description of the discrete dynamics can be given for the extended affine Weyl group $\widetilde W$ as well. 

\begin{rem}
Up to the author's knowledge, examples of non-abelian discrete systems of $ell$-type have not appeared yet. However, systems of $q$-type might be found in \cite{bobrova2023non}, where the non-commutative analogs of the $q$-P$_1$ and $q$-P$_2$ hierarchies are presented. Moreover, examples of non-abelian discrete $d$-systems can be found, for instance, in \cite{cassatella2014singularity} and \cite{adler2020}.
\end{rem}

\section{Systems of \texorpdfstring{$A_{n}^{(1)}$, $n \geq 2$}{A(n), n >= 2} type}
\label{sec:Ansys}
Let us present a non-abelian version of an important example that is closely connected with the dressing chains \cite{veselov1993}. Furthermore, this example is related to the Painlevé equations \cite{adler1994nonlinear} and arises from a generalisation of the symmetries for the P$_4$ and P$_5$ systems \cite{noumi2000affine}. All results of this section might have been considered as a generalisation of some results of the papers \cite{noumi2000affine} and \cite{nagoya2004quantum}.

According to the previous section, one needs to define a representation of the affine Weyl group. Let the Cartan matrix $C$ be of type $A_n^{(1)}$, $ n \geq 2$ and $I = \{0, 1, \dots, n\}$. Note that the matrix $C$ is as below and consider an $(n + 1) \times (n + 1)$-matrix $U$:

\begin{align}
    C
    &= (c_{i, j})
    = 
    \begin{pmatrix}
        2 & - 1 & 0 & \dots & 0 & -1 \\[0.9mm]
        -1 & 2 & -1 & \dots & 0 & 0 \\[0.9mm]
        0 & -1 & 2 & \dots & 0 & 0 \\[0.9mm]
        \vdots & \vdots & \vdots & \ddots & \vdots & \vdots \\[0.9mm]
        0 & 0 & 0 & \dots & 2 & -1 \\[0.9mm]
        -1 & 0 & 0 & \dots & -1 & 2
    \end{pmatrix},
    &
    U
    &= (u_{i, j})
    = 
    \begin{pmatrix}
        0 & 1 & 0 & \dots & 0 & -1 \\[0.9mm]
        -1 & 0 & 1 & \dots & 0 & 0 \\[0.9mm]
        0 & -1 & 0 & \dots & 0 & 0 \\[0.9mm]
        \vdots & \vdots & \vdots & \ddots & \vdots & \vdots \\[0.9mm]
        0 & 0 & 0 & \dots & 0 & 1 \\[0.9mm]
        1 & 0 & 0 & \dots & -1 & 0
    \end{pmatrix}
    .
\end{align}
The matrix $U$ is helpful in order to define the $s$-actions on the $f$-variables, i.e. we set
\begin{align}
    s_i (f_j)
    &= f_j + u_{i, j} \, \alpha_i \, f_i^{-1}.
\end{align}
It is easy to verify the following theorem that is a generalisation to the non-abelian case of its quantum version derived in the paper \cite{nagoya2004quantum}. Note that we do not assume any restrictions for the $f$-variables.

\begin{thm}
\label{thm:Aldsys}
Let us set
\begin{align}
    &&
    s_i(\alpha_i)
    &= - \alpha_i,
    &
    &&
    s_i (\alpha_j)
    &= \alpha_j + \alpha_i
    && 
    (j = i \pm 1),
    &&
    &
    &&
    s_i (\alpha_j)
    &= \alpha_j
    &&
    (j \neq i \pm 1),
    &&
    &&
    \\[1mm]
    &&
    s_i(f_i)
    &= f_i,
    &
    &&
    s_i (f_j)
    &= f_j \pm \alpha_i \, f_i^{-1}
    &&
    (j = i \pm 1),
    &&
    &
    &&
    s_i (f_j)
    &= f_j
    && 
    (j \neq i \pm 1),
    &&
    &&
    \\[1mm]
    &&
    \pi (\alpha_j)
    &= \alpha_{j + 1},
    &
    &&
    \pi (f_j)
    &= f_{j + 1},
    &
    j 
    &\in \quot{\mathbb{Z}}{(n + 1) \, \mathbb{Z}}
    .
\end{align}
The latter defines a birational representation of the extended affine Weyl group of type $A_n^{(1)}$, $n \geq 2$.
\end{thm}

\begin{rem}
When $n = 1$, such a representation is related to Bäcklund transformations of the second Painlevé equation. Its non-abelian extension can be found in Subsection \ref{sec:dPE7}. 
\end{rem}

Note that the shift operators are given by
\begin{align}
    &&
    T_1
    &= \pi \, s_n \, s_{n - 1} \, \dots \, s_1,
    &
    T_2
    &= s_1 \, \pi \, s_n \, \dots \, s_2,
    &
    \dots,&
    &
    T_{n + 1}
    &= s_n \, \dots \, s_1 \, \pi
    &&
\end{align}
and satisfy the relation  
\begin{align}
    T_1 \,\, T_2 \,\, \dots \,\, T_{n + 1} 
    &= 1.
\end{align}
Thus, any $n$ of them form a basis for the lattice. This representation for $n = 2$ and $n = 3$ in terms of other variables leading to the Hamiltonian form are discussed in Subsections \ref{sec:dPE6} and \ref{sec:dPD5} respectively. 

\medskip
Recall that Theorem \ref{thm:Aldsys} defines Bäcklund transformations for higher order symmetric forms of the \ref{eq:P40sys} and \ref{eq:P5sys} systems respectively. In the commutative case, these systems are closely connected with the dressing chains introduced in \cite{veselov1993} (see also \cite{takasaki2003spectral}). Thanks to the previous theorem, we have a generalisation of these systems to the non-abelian case. 

\begin{thm}
\label{thm:Ansys}
Let $j \in \quot{\mathbb{Z}}{(n + 1)\mathbb{Z}}$. Consider the systems for $n = 2 l$ and $n = 2 l + 1$
\begin{align}
    \label{eq:A2l_nc}
    \tag*{$A_{2l}^{(1)}$}
    &\begin{aligned}
    \dot f_j
    &= \sum_{1 \leq r \leq l} f_j \, f_{j + 2 r - 1}
    - \sum_{1 \leq r \leq l} f_{j + 2 r} \, f_j
    + \alpha_j;
    \end{aligned}
    \\[2mm]
    \label{eq:A2l1_nc}
    \tag*{$A_{2l + 1}^{(1)}$}
    &\begin{aligned}
    \tfrac12 \, t \, \dot f_j
    = \sum_{1 \leq r \leq s \leq l} f_j \, f_{j + 2 r - 1} \, f_{j + 2 s}
    &- \sum_{1 \leq r \leq s \leq l} f_{j + 2 r} \, f_{j + 2 s + 1} \, f_j
    \\[1mm]
    &+ \, \Big({
    \tfrac12 - \sum_{1 \leq r \leq l} \alpha_{j + 2 r}
    }\Big) f_j
    + \alpha_j \sum_{1 \leq r \leq l} f_{j + 2 r}
    .
    \end{aligned}
\end{align}
Then transformations given in Theorem {\rm\ref{thm:Aldsys}} are Bäcklund transformations of the \ref{eq:A2l_nc} and \ref{eq:A2l1_nc} systems.
\end{thm}
\begin{proof}
Can be done just by a direct computation. 
\end{proof}

\begin{rem}
When $l = 1$, the systems reduce to the \ref{eq:P40sys} and \ref{eq:P5sys} systems written in the $q$, $p$ variables. Similar to the commutative case, the order of the \ref{eq:A2l_nc} and \ref{eq:A2l1_nc} systems can be decreased thanks to the first integrals
\begin{align}
    &&&&
    \sum_{0 \leq j \leq n} f_j
    &= t,
    &&&
    \sum_{0 \leq r \leq l} f_{2 r}
    &= \sum_{0 \leq r \leq l} f_{2 r + 1}
    = \tfrac12 t
    &&&&
\end{align}
respectively. Note that $t$ might belong to $\mathcal{R}$. Such a fully non-abelian version of the P$_4$ system, i.e. the case of $A_2^{(1)}$, was derived in \cite{bobrova2022fully}.
\end{rem}

\medskip
Since the commutative systems \ref{eq:A2l_nc} and \ref{eq:A2l1_nc} are connected with the celebrating dressing chains, we would like to present their Lax pairs of the Zakharov-Shabat type in the non-abelian case. Namely, we have

\begin{thm}
\label{thm:Anpair}
Let $\Psi = \Psi (\lambda, t) \in \Mat_{n + 1}(\mathcal{R})$, $\lambda \in \mathcal{Z}(\mathcal{R})$ satisfy the linear system
\begin{align}
    \label{eq:linsys}
    \left\{
    \begin{array}{rcl}
         \partial_{\lambda} \Psi(\lambda, t)
         &=& \mathcal{A}(\lambda, t) \, \Psi(\lambda, t),
         \\[2mm]
         \partial_{t} \Psi (\lambda, t)
         &=& \mathcal{B}(\lambda, t) \, \Psi(\lambda, t),
    \end{array}
    \right.
\end{align}
where matrices $\mathcal{A} = \mathcal{A}(\lambda, t)$ and $\mathcal{B} = \mathcal{B}(\lambda, t)$ belong to $\Mat_{n + 1} (\mathcal{R})$ and depend on the spectral parameter $\lambda$ as
\begin{align}
    \mathcal{A}(\lambda)
    &= A_0 + A_{-1} \, \lambda^{-1},
    &
    \mathcal{B}(\lambda)
    &= B_1 \, \lambda + B_0
\end{align}
with the following matrix coefficients expressed in terms of the standard unit matrices $E_{r,s} \in \Mat_{n + 1} (\mathbb{C})$
\begin{gather}
    \begin{aligned}
    A_0
    &= E_{1, n} 
    + f_0 \, E_{1, n + 1} + E_{2, n + 1},
    &&&&&
    A_{-1}
    &= \sum_{1 \leq r \leq n + 1} \beta_r \, E_{r,r}
    + \sum_{1 \leq r \leq n} f_r \, E_{r + 1, r}
    + \sum_{1 \leq r \leq n - 1} E_{r + 2, r},
    \end{aligned}
    \\[2mm]
    \begin{aligned}
    B_1
    &= E_{1, n + 1},
    &&&&&
    B_0
    &= \sum_{1 \leq r \leq n + 1} g_{r} \, E_{r,r}
    + \sum_{1 \leq r \leq n} E_{r + 1, r}
    \end{aligned}
\end{gather}
and \,\, $\alpha_0 = 1 + \beta_{n + 1} - \beta_1$, \,\, $\alpha_j = \beta_j - \beta_{j + 1}$, \, $j \in \quot{\mathbb{Z}}{(n + 1)\mathbb{Z}} \setminus \{0\}$.
Then, there exists a set of the $g$-functions such that the compatibility condition of system \eqref{eq:linsys}, i.e. 
\begin{align}
    \label{eq:zcr}
    \partial_t \mathcal{A} - \partial_{\lambda} \mathcal{B}
    &= \mathcal{B} \, \mathcal{A} - \mathcal{A} \, \mathcal{B},
\end{align}
is equivalent to either the \ref{eq:A2l_nc} or \ref{eq:A2l1_nc} system. Their explicit forms are given by \eqref{eq:gsol_evenn} and \eqref{eq:gsol_oddn} respectively.
\end{thm}
\begin{proof}
It is enough to specify the form of the $g$-functions. This form depends on the system type. The compatibility condition \eqref{eq:zcr} yields the following constraints
\begin{align}
    \label{eq:compcond1}
    &&
    \dot f_j
    = - f_{j} \, g_j + g_{j + 1} \, f_{j} + \alpha_j,&
    &
    \\[1mm]
    \label{eq:compcond2}
    &&
    f_{j} - f_{j + 1}
    = g_{j} - g_{j + 2},&
    &
    j
    &\in \quot{\mathbb{Z}}{(n + 1) \mathbb{Z}}.
    &&
\end{align}
Without the periodicity condition, the system \eqref{eq:compcond2} can be solved as follows. Let us take the sum over $r = 0, \dots, j$, where $j \geq 0$:
\begin{align}
    &&
    \sum_{0 \leq r \leq j}
    (f_{r} - f_{r + 1})
    &= \sum_{0 \leq r \leq j} g_{j} - g_{j + 2}
    &
    \Rightarrow
    &&
    f_0 - f_{j + 1}
    &= g_0 + g_1
    - g_{j + 1} - g_{j + 2}.
    &&
\end{align}
Now we introduce the function $h_r = (-1)^r \, g_r$ and, thus, the latter can be rewritten as
\begin{align}
    h_{r + 2} - h_{r + 1}
    &= (-1)^r \brackets{
    g_0 + g_1 - f_0 + f_{r + 1}
    }.
\end{align}
After the summing over $r = 0, \dots, j$, we obtain
\begin{align}
    &&
    h_{j + 2} - h_{1}
    &= \sum_{0 \leq r \leq j} (-1)^r \brackets{
    g_0 + g_1 - f_0 + f_{r + 1}
    }
    \\[1mm]
    \label{eq:gsol_gen}
    \Leftrightarrow
    &&
    g_{j}
    &= (-1)^{j + 1} g_1
    + \sum_{0 \leq r \leq j - 2} (-1)^{r + j} \brackets{
    g_0 + g_1 - f_0 + f_{r + 1}
    }.
\end{align}
For even $j = 2k$ and odd $j = 2 k + 1$ values of $j$, the formula \eqref{eq:gsol_gen} is simply given by
\begin{align}
    \label{eq:gsol_evenodd}
    g_{2 k}
    &= g_0 + \sum_{0 \leq r \leq 2 k - 1} (-1)^{r + 1} f_{r},
    &
    g_{2 k + 1}
    &= g_1 + \sum_{0 \leq r \leq 2 k - 1} (-1)^{r + 1}  f_{r + 1}.
\end{align}
\begin{rem}
The system \eqref{eq:compcond2} might be rewritten in a matrix form with the help of the $U$-matrix, i.e. $U \, G = F$, where $G = \begin{pmatrix} g_1 & g_2 & \dots & g_n \end{pmatrix}^{T}$ and $F = \begin{pmatrix} f_1 - f_0 & f_2 - f_1 & \dots & f_0 - f_n \end{pmatrix}^{T}$.
\end{rem}

Turning to the periodicity condition $g_{n + 1} = g_0$, we need to distinguish the odd and even cases. When $n = 2 l$, we have $g_{2 l + 1} = g_0$ and, thus, either \eqref{eq:gsol_gen} or \eqref{eq:gsol_evenodd} leads to
\begin{align}
    g_1
    &= g_0 + \sum_{0 \leq r \leq 2 l - 1} (-1)^r f_{r + 1}.
\end{align}
Similarly, for $n = 2 l + 1$, one can obtain the following constraint on the $f$-functions:
\begin{align}
    \label{eq:cond3}
    \sum_{0 \leq r \leq l} f_{2 r}
    &= \sum_{0 \leq r \leq l} f_{2 r + 1}.
\end{align}
Therefore, when $n = 2 l$, the space of $g$-solutions is one-dimensional, while for $n = 2 l + 1$ it is two-dimensional and yields condition \eqref{eq:cond3}.
Since we have already known the form of the systems \ref{eq:A2l_nc} and \ref{eq:A2l1_nc}, the $g$-functions can be determined uniquely. Let us consider these two cases separately.

\medskip
\textbullet \,\, When $n = 2 l$, $g_1$ can be expressed in terms of $g_0$ and $f_j$. In this case the $g$-functions are given by 
\begin{align}
    \label{eq:gsol_even}
    &&
    g_{2 k}
    &= g_0 
    + \sum_{0 \leq r \leq 2 k - 1} (-1)^{r + 1} f_r,
    &
    g_{2k + 1}
    &= g_0 
    + \sum_{2 k \leq r \leq 2 l - 1} (-1)^{r} f_{r + 1}
    ,
    &
    k 
    &= 0, \dots, l.
    &&
\end{align}
In order to find the $g_0$-function, we substitute this solution into \eqref{eq:compcond1} and, then, require that it coincides with the~\ref{eq:A2l_nc} system. This leads to an explicit expression for the $g_0$-function:
\begin{align}
    g_0
    &= - \sum_{1 \leq r \leq l} f_{2 r - 1},
\end{align}
and, finally, we obtain an explicit form of the $g_j$-functions linearly expressed in terms of the $f$-functions:
\begin{align}
    \label{eq:gsol_evenn}
    g_j
    &= - \sum_{1 \leq r \leq l} f_{j + 2 r},
\end{align}
where indexes belong to the periodic lattice $\quot{\mathbb{Z}}{(n + 1) \mathbb{Z}}$.

\medskip
\textbullet \,\, As we have already mentioned, for $n = 2 l + 1$, the space of $g$-solutions for \eqref{eq:compcond2} is two-dimensional and the constraint \eqref{eq:cond3} holds. Note that the \ref{eq:A2l1_nc} system leads to the conditions
\begin{align}
    &&
    t \, \sum_{0 \leq r \leq l} \dot f_{2 r}
    &= \sum_{0 \leq r \leq l} f_{2 r},
    &
    t \, \sum_{0 \leq r \leq l} \dot f_{2 r + 1}
    &= \sum_{0 \leq r \leq l} f_{2 r + 1}.
    &&
\end{align}
They have the solution
\begin{align}
    \label{eq:fsol}
    &&
    \sum_{0 \leq r \leq l} f_{2 r}
    &= \gamma \, t,
    &
    \sum_{0 \leq r \leq l} f_{2 r + 1}
    &= \tilde \gamma \, t.
    &&
\end{align}
Let us use the normalisation $\gamma = \tilde \gamma = \tfrac12$ and notice that
\begin{align}
    \label{eq:cond4}
     \sum_{0 \leq r \leq l} \dot f_{2 r}
     &= \sum_{0 \leq r \leq l} \dot f_{2 r + 1}
     = \tfrac12.
\end{align}
The condition \eqref{eq:cond4} is crucial for determining the $g$-functions. 

Now we have all the necessary knowledge in our hands and, thus, can proceed to solving \eqref{eq:compcond1} and \eqref{eq:compcond2}. 
Let us substitute \eqref{eq:gsol_evenodd} and \eqref{eq:compcond1} into \eqref{eq:cond4}. The result can be written as
\begin{align}
    \sum_{0 \leq r \leq l} \brackets{
    - f_{2 r} \, g_0
    + g_1 \, f_{2 r}
    }
    &= \sum_{0 \leq r \leq l} \, \sum_{0 \leq s \leq 2 r - 1} (-1)^{s + 1} \brackets{
    f_{2 r} \, f_s 
    - f_{s + 1} \, f_{2 r}
    }
    + \Big( 
    \tfrac12 - \sum_{0 \leq r \leq l} \alpha_{2 r}
    \Big),
\end{align}
or, recalling \eqref{eq:fsol} and {assuming $t \in \mathcal{Z} (\mathcal{R})$},
\begin{align}
    &&
    g_1
    &= g_0 
    + 2 t^{-1} F
    + t^{-1} \, \Big( 
    1 - 2 \sum_{0 \leq r \leq l} \alpha_{2 r}
    \Big),
    &
    F
    &:= \sum_{0 \leq r \leq l} \, \sum_{0 \leq s \leq 2 r - 1} (-1)^{s + 1} \brackets{
    f_{2 r} \, f_s 
    - f_{s + 1} \, f_{2 r}
    }.
    &&
\end{align}
An explicit form of the $g_0$-function can be found by requiring that the condition \eqref{eq:compcond1} is equivalent to the \ref{eq:A2l1_nc} system. The resulting expressions for the $g_0$ and $g_1$ functions are given by\footnote{The author was not succeed in finding a solution for $t \in \mathcal{R}$.}
\begin{multline}
    g_0
    = - \sum_{0 \leq r \leq l - 1} f_{2 r + 1}
    + 2 \, t^{-1} \, 
    \sum_{1 \leq r \leq l} \, 
    \sum_{0 \leq s < r \leq l - 1}
    f_{2 r - 1} \, 
    f_{2 s},
    \\[1mm]
    g_1
    = - \tfrac12 \, t \, + \, f_0 \, + \sum_{0 \leq r \leq l - 1} f_{2 r + 1}
    - 2 \, t^{-1} \,
    \sum_{1 \leq r \leq l} \, 
    \sum_{0 \leq s < r \leq l - 1}
    f_{2 s} \, f_{2 r - 1}
    + t^{-1} \, \Big( 
    1 - 2 \sum_{0 \leq r \leq l} \alpha_{2 r}
    \Big)
    ,
\end{multline}
where we have used condition \eqref{eq:cond3}. 
So, a final form of the $g$-functions reads as 
\begin{align}
    \notag
    g_{2 k}
    = - \sum_{k \leq r \leq l - 1} f_{2 r + 1}
    - \sum_{0 \leq r \leq k - 1} f_{2 r}
    + 2 \, t^{-1} \, 
    \sum_{1 \leq r \leq l} \, 
    \sum_{0 \leq s < r \leq l - 1}
    f_{2 r - 1} \, 
    f_{2 s},
    \hspace{5cm}
    \\[1mm]
    \label{eq:gsol_oddn}
    g_{2 k + 1}
    = - \tfrac12 \, t 
    + \sum_{k \leq r \leq l - 1} f_{2 r + 1}
    + \sum_{0 \leq r \leq k} f_{2 r}
    - 2 \, t^{-1} \,
    \sum_{1 \leq r \leq l} \, 
    \sum_{0 \leq s < r \leq l - 1}
    f_{2 s} \, f_{2 r - 1}
    + t^{-1} \, \Big( 
    1 - 2 \sum_{0 \leq r \leq l} \alpha_{2 r}
    \Big),
    \\[1mm]
    \notag
    k
    = 0, 1, \dots, l
    .
\end{align}
Unlike the previous case $n = 2 l$, they are no longer linear in $f$, since $t$ is defined by \eqref{eq:fsol} with $\gamma = \tilde \gamma = \tfrac12$.
\end{proof}

\section{\texorpdfstring{$d$}{d}-\Painleve equations}
\label{sec:dP}
In this Section we use the following agreements. The variables $q$, $p$ belong to $\mathcal{R}$ and all constant parameters labeling by greek letters are from the field $\mathbb{C}$. Usually, $t$ is a central element, i.e. $t \in \mathcal{Z}(\mathcal{R})$, except for the \ref{eq:P20sys} and \ref{eq:P40sys} systems. For the discrete dynamics, we will use the usual notation. Namely, for a $T$-action of the translation operator $T$, we set $T(f) = \bar f$ and $T^{-1} (f) = \underline{f}$. Regarding the difference systems, we use $T^n (f) = f_n$. 

\begin{rem}
All computations have been done with the help of \href{https://mathweb.ucsd.edu/~ncalg/}{NCAlgebra package} for \href{https://www.wolfram.com/mathematica/}{Wolfram Mathematica}.
\end{rem}

\begin{rem}
Note that symmetries of the non-abelian analog of the first Painlevé equation written below do not form any affine Weyl group structure, however, unlike the commutative case, it has the $\tau$-symmetry. 
\begin{align}
    \tag*{P$_1$}
    \label{eq:P1sys}
    &
    \left\{
    \begin{array}{rcl}
         \dot q
         &=& p,  
         \\[2mm]
         \dot p
         &=& 6 q^2 + t.
    \end{array}
    \right.
\end{align}
\end{rem}

\begin{rem}
For some of the difference equations we present their continuous limits, which can be done as follows. For the simplicity, suppose that we have a difference equation for the functions $f_n$. One can take the change of variables with the commutative parameter $\varepsilon$
\begin{align}
    &&
    z
    &= \varepsilon \, n
    &&
\end{align}
supplemented by the maps
\begin{align}
    &&
    f_n
    &= F,
    &
    f_{n + k}
    &= F 
    + \, k \, \varepsilon \dot F
    + \, \tfrac12 k^2 \, \varepsilon^2 \ddot F
    + O (\varepsilon^3)
    .
    &&
\end{align}
The latter must be chosen in such a way that the limit $\varepsilon \to 0$ exists. Below we will give explicit examples, where in the formulas the capital letters correspond to the differential equations.
\end{rem}

\subsection{\texorpdfstring{$d$-P$(E_7)$}{dPE7}}
\label{sec:dPE7}
Consider the \ref{eq:P20sys} system \cite{Retakh_Rubtsov_2010}
\begin{align}
    \tag*{P$_2$}
    \label{eq:P20sys}
    &
    \left\{
    \begin{array}{rcr}
         \dot q
         &=& - q^2 + p - \tfrac12 t,  
         \\[2mm]
         \dot p
         &=& q p + p q + \alpha_1.
    \end{array}
    \right.
\end{align}
Here we assume that $t$ is also an element of $\mathcal{R}$ such that $\dot{t} = 1$. Let $\alpha_0 + \alpha_1 = 1$ and $f := - p + 2 q^2 + t$. Then, one can verify that Bäcklund transformations (BT) for this system are given in Table \ref{tab:BT_P20} (cf. with \cite{bershtein2023hamiltonian})
\begin{table}[H]
    \centering
    \begin{tabular}{c||cc|ccc}
         & $\alpha_0$
         & $\alpha_1$
         & $q$
         & $p$
         & $t$
         \\[1mm]
         \hline\hline
         $s_0$
         & $- \alpha_0$
         & $\alpha_1 + 2 \alpha_0$
         & $\,\, q - \alpha_0 \, f^{-1\phantom{\frac{1}{1}}}$
         & $p - 2 \alpha_0 q f^{-1} - 2 \alpha_0 f^{-1} q + 2 \alpha_0 f^{-2}$
         & $t$
         \\[1mm]
         $s_1$
         & $\alpha_0 + 2 \alpha_1$
         & $- \alpha_1$
         & $q + \alpha_1 p^{-1}$
         & $p$
         & $t$
         \\[1mm]
         \hline
         $\pi$
         & $\alpha_1$
         & $\alpha_0$
         & $- q$
         & $- p + 2 q^2 + t$
         & $t$
    \end{tabular}
    \caption{BT for the \ref{eq:P20sys} system}
    \label{tab:BT_P20}
\end{table}
\begin{proof}
Let us illustrate the computations in this simple case, because in the other cases they are so tedious. Thus, we have
\begin{align}
    \dot{s_1(q)}
    &= \dot q - \alpha_1 p^{-1} \dot p p^{-1}
    = \brackets{
    - q^2 + p - \tfrac12 t 
    }
    - \alpha_1 p^{-1} \brackets{
    q p + p q + \alpha_1
    } p^{-1}
    \\
    &= - \brackets{
    q^2 + \alpha_1 p^{-1} q + \alpha_1 q p^{-1} + \alpha_1^2 p^{-2}
    }
    + p - \tfrac12 t
    = - s_1(q)^2 + s_1(p) - \tfrac12 s_1(t),
    \\[2mm]
    \dot{s_1(p)}
    &= \dot p
    = q p + p q + \alpha_1
    = \brackets{q + \alpha_1 p^{-1}} p
    + p \brackets{q + \alpha_1 p^{-1}} - \alpha_1
    = s_1(q) s_1(p) + s_1(p) s_1(q) + s_1(\alpha_1);
    \\[2mm]
    \dot{\pi(q)}
    &= - \dot q
    = q^2 - p - \tfrac12 t
    = - q^2 + \brackets{2 q^2 - p + t} - \tfrac12 t
    = - \pi(q)^2 + \pi(p) - \tfrac12 \pi(t),
    \\[2mm]
    \dot{\pi(p)}
    &= - \dot p + 2 \dot q q + 2 q \dot q + 1
    = - (q p + p q + \alpha_1)
    + 2 \brackets{
    - q^2 + p - \tfrac12 t
    } q
    + 2 q \brackets{
    - q^2 - p + t
    }
    + 1 - \alpha_1
    \\
    &= - \brackets{
    2 q^2 - p + t
    } q
    - q \brackets{
    2 q^2 - p + t
    }
    + (1 - \alpha_1)
    = \pi(p) \pi(q) + \pi(q) \pi(p) + \pi(\alpha).
\end{align}
For the remaining element $s_0$, we note that $s_0 = \pi s_1 \pi$.
\end{proof}
Similar to the commutative case, these elements form an extended affine Weyl group of type $A_1^{(1)}$:
\begin{gather}
    \label{eq:WA1}
    \begin{gathered}
    \widetilde{W}(A_1^{(1)})
    = \angleBrackets{s_0, s_1; \pi},
    \\[1mm]
    \begin{aligned}
    s_i^2
    &= 1,
    &&&&&
    \pi^2 
    &= 1,
    &&&&&
    \pi s_i
    &= s_{i + 1} \pi,
    &&&&&
    i 
    &\in \quot{\mathbb{Z}}{2 \mathbb{Z}}.
    \end{aligned}
    \end{gathered}
\end{gather}
The corresponding Cartan matrix $C = (c_{ij})$ and Dynkin diagram are given below. 
\begin{table}[H]
    \vspace{-8mm}
    \begin{tabular}{ccc}
    $
    C
    = 
    \begin{pmatrix}
    \,\, 2 & - 2 \\[0.9mm] - 2 & \,\, 2    
    \end{pmatrix}
    $
    & \hspace{2cm} &
    \begin{tabular}{c}
    \\[0mm]
    \scalebox{0.9}{\tikzset{every picture/.style={line width=0.75pt}} 

\begin{tikzpicture}[x=0.75pt,y=0.75pt,yscale=-1,xscale=1]

\draw  [draw opacity=0][dash pattern={on 0.84pt off 2.51pt}] (80.04,121.45) .. controls (101.98,88.62) and (146.38,79.78) .. (179.21,101.71) .. controls (186.36,106.48) and (192.37,112.32) .. (197.17,118.87) -- (139.5,161.17) -- cycle ; \draw  [dash pattern={on 0.84pt off 2.51pt}] (80.04,121.45) .. controls (101.98,88.62) and (146.38,79.78) .. (179.21,101.71) .. controls (186.36,106.48) and (192.37,112.32) .. (197.17,118.87) ;  
\draw    (200.08,122.86) ;
\draw [shift={(200.93,124.11)}, rotate = 235.64] [color={rgb, 255:red, 0; green, 0; blue, 0 }  ][line width=0.75]    (6.56,-2.94) .. controls (4.17,-1.38) and (1.99,-0.4) .. (0,0) .. controls (1.99,0.4) and (4.17,1.38) .. (6.56,2.94)   ;
\draw  [fill={rgb, 255:red, 0; green, 0; blue, 0 }  ,fill opacity=1 ] (74.06,129.83) .. controls (74.06,127.39) and (76.04,125.41) .. (78.48,125.41) .. controls (80.91,125.41) and (82.89,127.39) .. (82.89,129.83) .. controls (82.89,132.26) and (80.91,134.24) .. (78.48,134.24) .. controls (76.04,134.24) and (74.06,132.26) .. (74.06,129.83) -- cycle ;
\draw  [fill={rgb, 255:red, 0; green, 0; blue, 0 }  ,fill opacity=1 ] (198.06,129.83) .. controls (198.06,127.39) and (200.04,125.41) .. (202.48,125.41) .. controls (204.91,125.41) and (206.89,127.39) .. (206.89,129.83) .. controls (206.89,132.26) and (204.91,134.24) .. (202.48,134.24) .. controls (200.04,134.24) and (198.06,132.26) .. (198.06,129.83) -- cycle ;
\draw [line width=0.75]    (91.89,128.33) -- (189.06,128.33)(91.89,131.33) -- (189.06,131.33) ;
\draw [shift={(198.06,129.83)}, rotate = 180] [fill={rgb, 255:red, 0; green, 0; blue, 0 }  ][line width=0.08]  [draw opacity=0] (10.72,-5.15) -- (0,0) -- (10.72,5.15) -- (7.12,0) -- cycle    ;
\draw [shift={(82.89,129.83)}, rotate = 360] [fill={rgb, 255:red, 0; green, 0; blue, 0 }  ][line width=0.08]  [draw opacity=0] (10.72,-5.15) -- (0,0) -- (10.72,5.15) -- (7.12,0) -- cycle    ;
\draw    (79.73,123.02) ;
\draw [shift={(79.2,124.19)}, rotate = 294.09] [color={rgb, 255:red, 0; green, 0; blue, 0 }  ][line width=0.75]    (6.56,-2.94) .. controls (4.17,-1.38) and (1.99,-0.4) .. (0,0) .. controls (1.99,0.4) and (4.17,1.38) .. (6.56,2.94)   ;

\draw (52.38,125.57) node [anchor=north west][inner sep=0.75pt]  [font=\large]  {$\alpha _{0}$};
\draw (212.38,125.55) node [anchor=north west][inner sep=0.75pt]  [font=\large]  {$\alpha _{1}$};
\draw (133.78,97.3) node [anchor=north west][inner sep=0.75pt]  [font=\large]  {$\pi $};

\end{tikzpicture}}
    \end{tabular}
    \end{tabular}
    \vspace{-10mm}
\end{table}
\hspace{-5.2mm}
Note that from these data it is easy to see how the root variables $\alpha_0$ and $\alpha_1$ change under the Weyl group action, since the reflection formula reads as \eqref{eq:refl}.

Proceeding to the discrete system, consider the translation operator $T = s_1 \pi$. It acts on the parameters according to the formula
\begin{align}
    T (\alpha_0, \, \alpha_1)
    &= (\alpha_0 - 1, \, \alpha_1 + 1),
\end{align}
while the $q$ and $p$ variables change as follows
\begin{align}
    \bar q
    &= s_1 \pi (q)
    = - s_1(q)
    = - q - \alpha_1 p^{-1},
    &
    \bar p
    &= s_1 \pi (p)
    = s_1(- p + 2 q^2 + t)
    = - p + 2 {\bar{q}}^2 + t.
\end{align}
So, we obtain the dynamics
\begin{align}
    \tag*{\ref{eq:dPE7}}
    T \brackets{
    q, \, p, \, t; \, \alpha_0, \, \alpha_1
    }
    &= \brackets{
    - q - \alpha_1 p^{-1}, \,
    - p + 2 {\bar{q}}^2 + t, \,
    t; \,
    \alpha_0 - 1, \,
    \alpha_1 + 1
    },
\end{align}
which is a non-commutative analog of the $d$-P$(E_7)$ equation from \cite{sakai2001rational} ($E_7^{(1)}$-surface on page 206). The latter has the \ref{eq:P1sys} system as a continuous limit given by the formulas
\begin{align}
    q
    &= 1 + \varepsilon^2 \, Q - \tfrac16 \, \varepsilon^3 \, P,
    &
    p
    &= - 2 + 2 \varepsilon^2 \, Q + \tfrac23 \, \varepsilon^3 \, P,
    &
    t
    &= - 6 + \tfrac13 \, \varepsilon^4 \, T,
    &
    \alpha_1
    &= 4 + \tfrac23 \, \varepsilon^4 \, T.
\end{align}

The \ref{eq:dPE7} system can be rewritten in the difference form
\begin{align}
    &&
    &\left\{
    \begin{array}{rcl}
         q_{n + 1} + q_n
         &=& - \alpha_{1, n} p_n^{-1}  
         \\[2mm]
         p_{n} + p_{n - 1}
         &=& 2 q_n^2 + t, 
    \end{array}
    \right.
    &
    \alpha_{1, n} 
    &= \alpha_1 + n,
    &&
\end{align}
which reduces to the following second-order difference equation: 
\begin{align}
    \tag*{alt-d-P$_1$}
    \label{eq:altdP1}
    &&
    \alpha_{1, n} \brackets{
    q_{n + 1} + q_n
    }^{-1}
    + \alpha_{1, n - 1} \brackets{
    q_n + q_{n - 1}
    }^{-1}
    &= - 2 q_n^2 - t,
    &
    \alpha_{1, n} 
    &= \alpha_1 + n.
    &&
\end{align}
We call it a non-commutative version of the alt-d-P$_1$ equation (see, e.g., the second equation in Appendix A.4 of the paper \cite{van2007discrete}).

\subsection{\texorpdfstring{$d$-P$(E_6)$}{dPE6}}
\label{sec:dPE6}
Take the \ref{eq:P40sys} system \cite{Bobrova_2022}
\begin{align}
    \tag*{P$_4$}
    \label{eq:P40sys}
    &
    \left\{
    \begin{array}{rcr}
         \dot q
         &=& - q^2 + q p + p q - t q - \alpha_1
         ,  
         \\[2mm]
         \dot p
         &=& - p^2 + q p + p q + p t + \alpha_2.
    \end{array}
    \right.
\end{align}
Recall that here $t \in \mathcal{R}$ such that $\dot{t} = 1$. Set $\alpha_0 + \alpha_1 + \alpha_2 = 1$ and $f := - p + q + t$ and consider the following Bäcklund transformations \cite{bobrova2022fully} (see also Theorem \ref{thm:Aldsys})
\begin{table}[H]
    \centering
    \begin{tabular}{c||ccc|ccc}
         & $\alpha_0$
         & $\alpha_1$
         & $\alpha_2$
         & $q$
         & $p$
         & $t$
         \\[1mm]
         \hline\hline
         $s_0$
         & $- \alpha_0$
         & $\alpha_1 + \alpha_0$
         & $\alpha_2 + \alpha_0$
         & $\,\, q - \alpha_0 f^{-1 \phantom{\frac{1}{}}}$
         & $p - \alpha_0 f^{-1}$
         & $t$
         \\[1mm]
         $s_1$
         & $\alpha_0 + \alpha_1$
         & $- \alpha_1$
         & $\alpha_2 + \alpha_1$
         & $q$
         & $p - \alpha_1 q^{-1}$
         & $t$
         \\[1mm]
         $s_2$
         & $\alpha_0 + \alpha_2$
         & $\alpha_1 + \alpha_2$
         & $- \alpha_2$
         & $q + \alpha_2 p^{-1}$
         & $p$
         & $t$
         \\[1mm]
         \hline
         $\pi$
         & $\alpha_1$
         & $\alpha_2$
         & $\alpha_0$
         & $- p$
         & $- p + q + t$
         & $t$
    \end{tabular}
    \caption{BT for the \ref{eq:P40sys} system}
    \label{tab:BT_P40}
\end{table}
Similar to the commutative case, they form an extended affine Weyl group of type $A_2^{(1)}$
\begin{gather}
    \label{eq:WA2}
    \begin{gathered}
    \widetilde{W}(A_2^{(1)})
    = \angleBrackets{s_0, s_1, s_2; \pi},
    \\[1mm]
    \begin{aligned}
    s_i^2
    &= 1,
    &&&&&
    (s_i \, s_{i + 1})^3 
    &= 1,
    &&&&&
    \pi^3 
    &= 1,
    &&&&&
    \pi s_i
    &= s_{i + 1} \pi,
    &&&&&
    i 
    &\in \quot{\mathbb{Z}}{3 \mathbb{Z}}.
    \end{aligned}
    \end{gathered}
\end{gather}
In order to illustrate the $\widetilde W$-action on the root variables, we present the Cartan matrix and Dynkin diagram
\begin{table}[H]
    \vspace{-5mm}
    \begin{tabular}{ccc}
    $C
    = 
    \begin{pmatrix}
    2 & - 1 & - 1 \\[0.9mm] 
    - 1 & 2 & - 1 \\[0.9mm]
    - 1 & - 1 & 2    
    \end{pmatrix}
    $
    & \hspace{2cm} &
    \begin{tabular}{c}
    \\[-1mm]
    \scalebox{0.85}{\tikzset{every picture/.style={line width=0.75pt}} 

\begin{tikzpicture}[x=0.75pt,y=0.75pt,yscale=-1,xscale=1]

\draw   (292.14,64) -- (354.14,170.16) -- (230.14,170.16) -- cycle ;
\draw  [dash pattern={on 0.84pt off 2.51pt}] (221,135.5) .. controls (221,96.01) and (253.01,64) .. (292.5,64) .. controls (331.99,64) and (364,96.01) .. (364,135.5) .. controls (364,174.99) and (331.99,207) .. (292.5,207) .. controls (253.01,207) and (221,174.99) .. (221,135.5) -- cycle ;
\draw  [fill={rgb, 255:red, 0; green, 0; blue, 0 }  ,fill opacity=1 ] (288.09,64) .. controls (288.09,61.56) and (290.06,59.59) .. (292.5,59.59) .. controls (294.94,59.59) and (296.91,61.56) .. (296.91,64) .. controls (296.91,66.44) and (294.94,68.41) .. (292.5,68.41) .. controls (290.06,68.41) and (288.09,66.44) .. (288.09,64) -- cycle ;
\draw  [fill={rgb, 255:red, 0; green, 0; blue, 0 }  ,fill opacity=1 ] (225.73,170.16) .. controls (225.73,167.72) and (227.71,165.75) .. (230.14,165.75) .. controls (232.58,165.75) and (234.56,167.72) .. (234.56,170.16) .. controls (234.56,172.6) and (232.58,174.57) .. (230.14,174.57) .. controls (227.71,174.57) and (225.73,172.6) .. (225.73,170.16) -- cycle ;
\draw  [fill={rgb, 255:red, 0; green, 0; blue, 0 }  ,fill opacity=1 ] (349.73,170.16) .. controls (349.73,167.72) and (351.71,165.75) .. (354.14,165.75) .. controls (356.58,165.75) and (358.56,167.72) .. (358.56,170.16) .. controls (358.56,172.6) and (356.58,174.57) .. (354.14,174.57) .. controls (351.71,174.57) and (349.73,172.6) .. (349.73,170.16) -- cycle ;
\draw    (358,164.6) ;
\draw [shift={(357.44,166)}, rotate = 291.9] [color={rgb, 255:red, 0; green, 0; blue, 0 }  ][line width=0.75]    (6.56,-2.94) .. controls (4.17,-1.38) and (1.99,-0.4) .. (0,0) .. controls (1.99,0.4) and (4.17,1.38) .. (6.56,2.94)   ;
\draw    (233.5,175.77) -- (233.42,175.66) ;
\draw [shift={(232.17,174.1)}, rotate = 51.34] [color={rgb, 255:red, 0; green, 0; blue, 0 }  ][line width=0.75]    (6.56,-2.94) .. controls (4.17,-1.38) and (1.99,-0.4) .. (0,0) .. controls (1.99,0.4) and (4.17,1.38) .. (6.56,2.94)   ;
\draw    (286.09,64.5) ;
\draw [shift={(287.33,64.52)}, rotate = 180.72] [color={rgb, 255:red, 0; green, 0; blue, 0 }  ][line width=0.75]    (6.56,-2.94) .. controls (4.17,-1.38) and (1.99,-0.4) .. (0,0) .. controls (1.99,0.4) and (4.17,1.38) .. (6.56,2.94)   ;
\draw    (292.5,64) -- (232.33,165.77) ;
\draw    (354.14,170.16) -- (294.67,68.43) ;
\draw    (230.14,170.16) -- (349.83,169.77) ;

\draw (204.04,165.9) node [anchor=north west][inner sep=0.75pt]  [font=\large]  {$\alpha _{1}$};
\draw (364.04,165.88) node [anchor=north west][inner sep=0.75pt]  [font=\large]  {$\alpha _{2}$};
\draw (286.24,45.34) node [anchor=north west][inner sep=0.75pt]  [font=\large,rotate=-0.45]  {$\alpha _{0}$};
\draw (346.62,102.66) node [anchor=north west][inner sep=0.75pt]  [font=\large,rotate=-59.99]  {$\pi $};
\draw (233.36,111.95) node [anchor=north west][inner sep=0.75pt]  [font=\large,rotate=-299.77]  {$\pi $};
\draw (286.31,190.33) node [anchor=north west][inner sep=0.75pt]  [font=\large,rotate=-359.84]  {$\pi $};

\end{tikzpicture}}
    \end{tabular}
    \end{tabular}
    \vspace{-5mm}
\end{table}
Turning to the discrete dynamics, we are going to consider two independent directions on the $(\alpha_0, \alpha_1, \alpha_2)$-lattice, namely
\begin{align}
    T
    &:= s_1 \pi s_2,
    &
    T (\alpha_0, \alpha_1, \alpha_2)
    &= (\alpha_0, \alpha_1 + 1, \alpha_2 - 1),
    \\[1mm]
    T'
    &:= s_2 s_1 \pi
    ,
    &
    T' (\alpha_0,  \alpha_1, \alpha_2)
    &= (\alpha_0 - 1, \alpha_1, \alpha_2 + 1).
\end{align}

\medskip
\subsubsection{{$T$-direction}}
\label{sec:dPE6_1}
The corresponding translation operator defined by $T = s_1 \pi s_2$ acts on the variables and parameters as follows
\begin{align}
    \tag*{\ref{eq:dPE6}}
    T \brackets{
    q, \, p, \, t; \, \alpha_0, \, \alpha_1, \, \alpha_2
    }
    &= \brackets{
    - t - q + \bar p - (\alpha_2 - 1) \, {\bar{p}}^{-1}, \,
    t + q - p + \alpha_1 \, q^{-1}, \, 
    t; \,
    \alpha_0, \,
    \alpha_1 + 1, \,
    \alpha_2 - 1
    }.
\end{align}
The latter is a non-abelian version of the discrete dynamic related to the $E_6^{(1)}$-surface (see page 205 in \cite{sakai2001rational}), which is also known as the d-P$_2$ equation. The \ref{eq:dPE6} can be rewritten as the second-order difference system
\begin{align}
    \tag*{d-P$_2$}
    \label{eq:dP2}
    &&
    &\left\{
    \begin{array}{rcr}
         q_{n} + q_{n - 1}
         &=&  - t + p_n - \alpha_{2, n} \, p_n^{-1} 
         ,
         \\[2mm]
         p_{n + 1} + p_{n}
         &=&  t + q_n + \alpha_{1, n} \, q_n^{-1}, 
    \end{array}
    \right.
    &
    \alpha_{1, n} 
    &= \alpha_1 + n,
    &
    \alpha_{2, n}
    &= \alpha_2 - n
    .
    &&
\end{align}
Its continuous limit defined by the formulas given below is the \ref{eq:P20sys} system:
\begin{gather}
\begin{aligned}
    q
    &= - 1 - \varepsilon \, Q - \tfrac14 \, \varepsilon^2 \, P,
    &&&
    p
    &= 1 - \varepsilon \, Q + \tfrac14 \, \varepsilon^2 \, P,
    &&&
    t
    &= 2 + \tfrac14 (\mathrm{A}_1 - 1) \, \varepsilon^3,
\end{aligned}
\\[1mm]
\begin{aligned}
    \alpha_1
    &= - 1 - \tfrac12 \, \varepsilon^2 \, T
    + \tfrac12 (\mathrm{A}_1 - 1) \, \varepsilon^3,
    &&&
    \alpha_2
    &= 1 + \tfrac12 \, \varepsilon^2 \, T.
\end{aligned}
\end{gather}
Note that, after a restriction of the parameter $\alpha_2$, the \ref{eq:dPE6} can be reduced to the difference equation
\begin{align}
    \tag*{d-P$_1$}
    \label{eq:dP1}
    &&
    q_{n + 1} + q_n + q_{n - 1}
    &= \alpha_n q_n^{-1} - t,
    &
    \alpha_n 
    &= \alpha_1 + n,
    &&
\end{align}
which is known as the d-P$_1$ equation in the commutative case (see, e.g., Appendix A.1 in \cite{van2007discrete}).

\begin{rem}
Regarding the matrix \ref{eq:dP1} equation, it passes the singularity confinement test \cite{cassatella2014singularity} which~is a discrete analog of the matrix Painlevé-Kovalevskaya test introduced in \cite{Balandin_Sokolov_1998}. 
\end{rem}
\begin{rem}
Another non-abelian (matrix) versions of the d-P$_1$ equation were obtained in the paper \cite{adler2020}, where the author had used a symmetry reduction of the non-abelian Volterra lattices. 
\end{rem}

\medskip
\subsubsection{{$T'$-direction}}
\label{sec:dPE6_2}
The $T'$-operator defined by $T' = s_2 s_1 \pi$ gives the dynamics
\begin{align}
    T' \brackets{
    q, \, p, \, t; \, \alpha_0, \, \alpha_1, \, \alpha_2
    }
    &= \brackets{
    - p + (\alpha_1 + \alpha_2) (\alpha_2 p^{-1} + q)^{-1}, \,
    t + \bar q + q + \alpha_2 p^{-1}, \, 
    t; \,
    \alpha_0 - 1, \,
    \alpha_1, \,
    \alpha_2 + 1
    },
\end{align}
which can be rewritten in the form 
\begin{align}
    \tag*{\ref{eq:dPE6'}}
    &&
    \bar f \, f
    &= - (\bar g - \alpha_1) (\bar g + \alpha_2)^{-1} \bar g,
    &
    \bar g
    + g
    + \LieBrackets{f \, \bar g, f^{-1}}
    &= \alpha_1
    + f t + f^2,
    &
    \bar \alpha_1
    &= \alpha_1,
    &
    \bar \alpha_2
    &= \alpha_2 + 1,
    &&
\end{align}
where $f := q$ and $\bar g := p q$. Its abelian analog is eq. (8.258) from \cite{kajiwara2017geometric}.
\begin{proof}
Indeed, we first note that
\begin{align}
    &&
    \bar q 
    &= (- p q + \alpha_1) (\alpha_2 + p q)^{-1} q,
    &
    \Leftrightarrow&
    &
    \bar q \, q
    &= - (p q - \alpha_1) (p q + \alpha_2)^{-1} p q,
    &&
\end{align}
Set $\bar g := p q$. Then, on the one hand, $g = \underline{p} \, \underline{q}$ and, on the other hand,
\begin{align}
    \underline{p}
    &= T_2^{-1} (p)
    = \alpha_0 (t - p + q)^{-1} - q,
    &
    \underline{q}
    &= T_2^{-1} (q)
    = (1 - \alpha_2) \underline{p}^{-1} - (t - p + q),
\end{align}
where $T_2^{-1} = \pi^{-1} s_1^{-1} s_2^{-1} = \pi^2 s_1 s_2$. Hence,
\begin{align}
    \underline{p} \, \underline{q}
    &= \underline{p} \brackets{
    (1 - \alpha_2) \underline{p}^{-1} - (t - p + q)
    }
    = 1 - \alpha_2
    - \brackets{
    \alpha_0 (t - p + 1)^{-1} - q
    } (t - p + q)
    = \alpha_1 + q t - q p + q^2,
\end{align}
or, equivalently,
\begin{align}
    g + \bar g + q p - p q
    &= \alpha_1 + q t + q^2.
\end{align}
Noting that 
\begin{align}
    q p - p q 
    = q p q \,\, q^{-1} - q^{-1} \,\, q p q = \LieBrackets{q \,\, p q, q^{-1}} 
    = \LieBrackets{q \, \bar{g}, q^{-1}}
\end{align}
and recalling the definition of $f$, we arrive at the wanted dynamics \ref{eq:dPE6'}.
\end{proof} 
Similar to the previous case, one can write down the related system of difference equations. It is given~by
\begin{align}
    &&
    &\left\{ 
    \begin{array}{rcl}
         f_{n} \, f_{n - 1} 
         &=& - (g_n - \alpha_{1, n}) \, 
         (g_n + \alpha_{2, n} - 1)^{-1} \,
         g_n
         ,
         \\[2mm]
         g_{n + 1} + g_{n} 
         + \LieBrackets{f_n \, g_{n + 1}, f_n^{-1}} 
         &=& \alpha_{1, n} 
         + f_n \, t + f_n^2, 
    \end{array}
    \right.
    &
    &\begin{aligned}
    \alpha_{1, n} 
    &= \alpha_{1},
    \\[0.8mm]
    \alpha_{2, n}
    &= \alpha_2 + n.
    \end{aligned}
    &&
\end{align}

\subsection{\texorpdfstring{$d$-P$(D_7)$}{dPD7}}
\label{sec:dPD7}
In the abelian case, this surface type is connected with the special case of the third \Painleve equation, the so-called P$_3(D_7)$ equation. Consider its non-abelian analog given by the \ref{eq:P37sys} system \cite{Kawakami_2015}
\begin{align}
    \tag*{P$_{3,7}$}
    \label{eq:P37sys}
    &
    \left\{
    \begin{array}{rcl}
         t \, \dot q
         &=& 2 q p q + \alpha_1 q + t
         ,  
         \\[2mm]
         t \, \dot p
         &=& - 2 p q p - \alpha_1 p - 1
         .
    \end{array}
    \right.
\end{align}
Here $t$ belongs to the center $\mathcal{Z}(\mathcal{R})$. Let $\alpha_0 + \alpha_1 = 1$. Bäcklund transformations are given in Table \ref{tab:BT_P37} (cf.~with \cite{bershtein2023hamiltonian}) and form an extended affine Weyl group of type $A_1^{(1)}$ (see \eqref{eq:WA1} for more details).

\begin{minipage}[l]{0.65\textwidth}
\begin{table}[H]
    \centering
    \begin{tabular}{c||cc|ccc}
         & $\alpha_0$
         & $\alpha_1$
         & $q$
         & $p$
         & $t$
         \\[1mm]
         \hline\hline
         $s_0$
         & $- \alpha_0$
         & $\alpha_1 + 2 \alpha_0$
         & $q$
         & $\,\, p - \alpha_0 q^{-1} + t q^{-2 \phantom{\frac{1}{}}}$
         & $- t$
         \\[1mm]
         $s_1$
         & $\alpha_0 + 2 \alpha_1$
         & $- \alpha_1$
         & $- q - \alpha_1 p^{-1} - p^{-2}$
         & $- p$
         & $- t$
         \\[1mm]
         \hline
         $\pi$
         & $\alpha_1$
         & $\alpha_0$
         & $t p$
         & $- t^{-1\phantom{\frac{1}{}}} q \,\,$
         & $- t$
    \end{tabular}
    \caption{BT for the \ref{eq:P37sys} system}
    \label{tab:BT_P37}
\end{table}
\end{minipage}
\hspace{0.5cm}
\begin{minipage}[l]{0.3\textwidth}
    \scalebox{0.9}{\tikzset{every picture/.style={line width=0.75pt}} 

\begin{tikzpicture}[x=0.75pt,y=0.75pt,yscale=-1,xscale=1]

\draw  [draw opacity=0][dash pattern={on 0.84pt off 2.51pt}] (80.04,121.45) .. controls (101.98,88.62) and (146.38,79.78) .. (179.21,101.71) .. controls (186.36,106.48) and (192.37,112.32) .. (197.17,118.87) -- (139.5,161.17) -- cycle ; \draw  [dash pattern={on 0.84pt off 2.51pt}] (80.04,121.45) .. controls (101.98,88.62) and (146.38,79.78) .. (179.21,101.71) .. controls (186.36,106.48) and (192.37,112.32) .. (197.17,118.87) ;  
\draw    (200.08,122.86) ;
\draw [shift={(200.93,124.11)}, rotate = 235.64] [color={rgb, 255:red, 0; green, 0; blue, 0 }  ][line width=0.75]    (6.56,-2.94) .. controls (4.17,-1.38) and (1.99,-0.4) .. (0,0) .. controls (1.99,0.4) and (4.17,1.38) .. (6.56,2.94)   ;
\draw  [fill={rgb, 255:red, 0; green, 0; blue, 0 }  ,fill opacity=1 ] (74.06,129.83) .. controls (74.06,127.39) and (76.04,125.41) .. (78.48,125.41) .. controls (80.91,125.41) and (82.89,127.39) .. (82.89,129.83) .. controls (82.89,132.26) and (80.91,134.24) .. (78.48,134.24) .. controls (76.04,134.24) and (74.06,132.26) .. (74.06,129.83) -- cycle ;
\draw  [fill={rgb, 255:red, 0; green, 0; blue, 0 }  ,fill opacity=1 ] (198.06,129.83) .. controls (198.06,127.39) and (200.04,125.41) .. (202.48,125.41) .. controls (204.91,125.41) and (206.89,127.39) .. (206.89,129.83) .. controls (206.89,132.26) and (204.91,134.24) .. (202.48,134.24) .. controls (200.04,134.24) and (198.06,132.26) .. (198.06,129.83) -- cycle ;
\draw [line width=0.75]    (91.89,128.33) -- (189.06,128.33)(91.89,131.33) -- (189.06,131.33) ;
\draw [shift={(198.06,129.83)}, rotate = 180] [fill={rgb, 255:red, 0; green, 0; blue, 0 }  ][line width=0.08]  [draw opacity=0] (10.72,-5.15) -- (0,0) -- (10.72,5.15) -- (7.12,0) -- cycle    ;
\draw [shift={(82.89,129.83)}, rotate = 360] [fill={rgb, 255:red, 0; green, 0; blue, 0 }  ][line width=0.08]  [draw opacity=0] (10.72,-5.15) -- (0,0) -- (10.72,5.15) -- (7.12,0) -- cycle    ;
\draw    (79.73,123.02) ;
\draw [shift={(79.2,124.19)}, rotate = 294.09] [color={rgb, 255:red, 0; green, 0; blue, 0 }  ][line width=0.75]    (6.56,-2.94) .. controls (4.17,-1.38) and (1.99,-0.4) .. (0,0) .. controls (1.99,0.4) and (4.17,1.38) .. (6.56,2.94)   ;

\draw (52.38,125.57) node [anchor=north west][inner sep=0.75pt]  [font=\large]  {$\alpha _{0}$};
\draw (212.38,125.55) node [anchor=north west][inner sep=0.75pt]  [font=\large]  {$\alpha _{1}$};
\draw (133.78,97.3) node [anchor=north west][inner sep=0.75pt]  [font=\large]  {$\pi $};

\end{tikzpicture}}
\end{minipage}
\begin{rem}
There is a version of the \ref{eq:P37sys} with a non-abelian constant \cite{bobrova2022classification}:
\begin{align}
    &
    \left\{
    \begin{array}{rcl}
         t \, \dot q
         &=& 2 q p q + \kappa_1 q + t \, h
         ,  
         \\[2mm]
         t \, \dot p
         &=& - 2 p q p - \kappa_1 p - \kappa_3
         .
    \end{array}
    \right.
\end{align}
When $h$ is commutative, it can be obtained from the \ref{eq:P37sys} by the scaling
\begin{align}
    &&
    q
    &\mapsto \kappa_3 \, q,
    &
    p 
    &\mapsto \kappa_3^{-1} p,
    &
    t 
    &\mapsto (\kappa_3 \, h)^{-1} t
    .
    &&
\end{align}
Since $\kappa_3$ is an inessential parameter, we set $\kappa_3 = 1$. The system with $h \in \mathcal{R}$ has the same Bäcklund transformations as in Table \ref{tab:BT_P37} except of the $s_0$- and $\pi$-elements, which now read as
\begin{align}
    s_0 \brackets{
    q, \, 
    p, \, 
    t; \, \alpha_0, \alpha_1
    }
    &= \brackets{
    h \, q \, h^{-1}, \,
    h \, (
    p - \alpha_0 \, q^{-1} + t \, q^{-1} \, h \, q^{-1}
    ) \, h^{-1}, \,
    - t; \,
    - \alpha_0, 
    \alpha_1 + 2 \alpha_0
    },
    \\[1mm]
    \pi \brackets{
    q, \, p, \, t; \, \alpha_0, \alpha_1
    }
    &= \brackets{
    t \, h \ p, \,
    - t^{-1} q \, h^{-1}, \,
    - t; \,
    \alpha_1, 
    \alpha_0
    }.
\end{align}
\end{rem}

Similar to Subsection \ref{sec:dPE7}, we consider the translation operator $T = \pi s_1$. Then, the $T$-dynamic reads as
\begin{align}
    T \brackets{
    q, \, p, \, t; \, \alpha_0, \alpha_1
    }
    &= \brackets{
    \alpha_0 \bar p^{-1} - \bar p^{-2} - t p, \,
    t^{-1} q, \,
    t; \,
    \alpha_0 + 1,
    \alpha_1 - 1
    }
\end{align}
and can be rewritten in the canonical form
\begin{align}
    \tag*{\ref{eq:dPD7}}
    \bar f + f 
    + \LieBrackets{g \, \bar f, g^{-1}}
    &= - \alpha_1 - t g^{-1},
    &
    \bar g \, g
    &= t \bar f,
    &
    \bar{\alpha}_1
    &= \alpha_1 - 1,
\end{align}
where $g := t p$ and $\bar{f} := q p$. Its commutative version is given by eqs. (2.37) -- (2.38) in \cite{sakai2007problem}.
\begin{proof}
Note that for the $p$-dynamic we have
\begin{align}
    &&
    \bar p
    &= t^{-1} q
    &
    \Leftrightarrow&
    &
    \bar{t p} \, (t p)
    &= t q p.
    &&
\end{align}
Setting $g := t p$ and $\bar{f} := q p$, it becomes $\bar {g} \, g = t \bar{f}$. Let us find $f = \underline{q} \, \underline{p}$. One can compute
\begin{align}
    \underline{q}
    &= t p, 
    &
    \underline{p}
    &= - \alpha_1 \underline{q}^{-1} 
    - t \underline{q}^{-2}
    - t^{-1} q,
\end{align}
then
\begin{align}
    \underline{q} \, \underline{p}
    &= \underline{q} \, \brackets{
    - \alpha_1 \underline{q}^{-1} 
    - t \underline{q}^{-2}
    - t^{-1} q
    }
    = - \alpha_1 - t \underline{q}^{-1} - t^{-1} \underline{q} \, q
    = - \alpha_1 - t g^{-1} - p q
    ,
\end{align}
or, equivalently,
\begin{align}
    \bar{f} + f
    + \LieBrackets{g \, \bar{f}, g^{-1}}
    &= - \alpha_1 - t g^{-1},
\end{align}
since $p q - q p = p q p \,\, p^{-1} - p^{-1} \,\, p q p = \LieBrackets{t p \,\, q p, t^{-1} p^{-1}} = \LieBrackets{g \, \bar{f}, g^{-1}}$. As a result, we obtain the \ref{eq:dPD7}.
\end{proof}

The related difference system
\begin{align}
    &&
    &\left\{
    \begin{array}{rcl}
         f_{n + 1} + f_{n}
         + \LieBrackets{g_n f_{n + 1}, g_n^{-1}}
         &=&  - \alpha_{1, n} 
         - t g_n^{-1},
         \\[2mm]
         g_{n} \, g_{n - 1}
         &=& t f_n, 
    \end{array}
    \right.
    &
    \alpha_{1, n} 
    &= \alpha_{1} - n
    &&
\end{align}
reduces to the second-order difference equation
\begin{align}
    \tag*{alt-d-P$_1'$}
    \label{eq:altdP1'}
    &&
    g_{n + 1} + g_{n - 1}
    &= - \alpha_{1, n} t \, g_n^{-1} - t^2 \, g_n^{-2},
    &
    \alpha_{1, n} 
    &= \alpha_1 - n.
    &&
\end{align}
The latter is a non-abelian version of one of the alt-d-P$_1$ equations (see, e.g., Appendix A.4 in \cite{van2007discrete}).

\begin{rem}
After the transformation
\begin{align}
    f
    &= - \tfrac12 + \varepsilon^2 \, Q - \tfrac16 \, \varepsilon^3 \, P,
    &
    g
    &= 1 - \varepsilon^2 \, Q - \tfrac13 \, \varepsilon^3 \, P,
    &
    t
    &= - 2 + \tfrac13 \, \varepsilon^4 \, T,
    &
    \alpha_1
    &= 3 - \tfrac23 \, \varepsilon^4 \, T,
\end{align}
the \ref{eq:dPD7} system becomes the \ref{eq:P1sys} in the limit $\varepsilon \to 0$.
\end{rem}

\subsection{\texorpdfstring{$d$-P$(D_6)$}{dPD6}}
\label{sec:dPD6}
One of the non-abelian analogs of the third \Painleve equation may be written in the form
\begin{align}
    \tag*{P$_{3,6}$}
    \label{eq:P36sys}
    &
    \left\{
    \begin{array}{rcl}
         t \, \dot q
         &=& 2 q p q - q^2 + (\alpha_1 + \beta_1) q + t
         ,  
         \\[2mm]
         t \, \dot p
         &=& - 2 p q p + q p + p q - (\alpha_1 + \beta_1) p + \alpha_1
         .
    \end{array}
    \right.
\end{align}
Recall that $t$ is also a central element. Let us set $\alpha_0 + \alpha_1 = 1$ and $\beta_0 + \beta_1 = 1$. Then the system has the following Bäcklund transformations (cf. with \cite{bershtein2023hamiltonian}), where $\tilde f := \pi s_1 \pi (f)$ and $\tilde f' := \pi' s_1' \pi' (f)$.
\begin{table}[H]
    \centering
    \begin{tabular}{c||cccc|ccc}
         & $\alpha_0$
         & $\alpha_1$
         & $\beta_0$
         & $\beta_1$
         & $q$
         & $p$
         & $t$
         \\[1mm]
         \hline\hline
         $s_0$
         & $- \alpha_0$
         & $\alpha_1 + 2 \alpha_0$
         & $\beta_0$
         & $\beta_1$
         & $\tilde q$
         & $\tilde p$
         & $t$
         \\[1mm]
         $s_1$
         & $\alpha_0 + 2 \alpha_1$
         & $- \alpha_1$
         & $\beta_0$
         & $\beta_1$
         & $q + \alpha_1 p^{-1}$
         & $p$
         & $t$
         \\[1mm]
         $s_0'$
         & $\alpha_0$
         & $\alpha_1$
         & $- \beta_0$
         & $\beta_1 + 2 \beta_0$
         & $\tilde q'$
         & $\tilde p'$
         & $t$
         \\[1mm]
         $s_1'$
         & $\alpha_0$
         & $\alpha_1$
         & $\beta_0 + 2 \beta_1$
         & $- \beta_1$
         & $q + \beta_1 (p - 1)^{-1}$
         & $p$
         & $t$
         \\[1mm]
         \hline
         $\pi$
         & $\alpha_1$
         & $\alpha_0$
         & $\beta_0$
         & $\beta_1$
         & $- t q^{-1 \phantom{\frac{1}{}}}$
         & $t^{-1} \brackets{q (p - 1) q + \beta_1 q}$
         & $t$
         \\[1mm]
         $\pi'$
         & $\alpha_0$
         & $\alpha_1$
         & $\beta_1$
         & $\beta_0$
         & $t q^{-1}$
         & $- t^{-1} \brackets{q p q + \alpha_1 q}$
         & $t$
         \\[1mm]
         \hline
         $\sigma$
         & $\beta_0$
         & $\beta_1$
         & $\alpha_0$
         & $\alpha_1$
         & $- q$
         & $1 - p$
         & $- t$
    \end{tabular}
    \caption{BT for the \ref{eq:P36sys} system}
    \label{tab:BT_P36}
\end{table}
\vspace{-2mm}
\hspace{-0.5cm}
\begin{minipage}[l]{0.65\textwidth}
Regarding $\alpha_0$, $\alpha_1$ and $\beta_0$, $\beta_1$ as the root variables, these Bäcklund transformations form an extended affine Weyl group of type $2 A_1^{(1)}$, where $\tilde \pi := \sigma \pi' \pi \sigma$:
\begin{gather}
    \label{eq:W2A1}
    \begin{gathered}
    \widetilde{W}(2 A_1^{(1)})
    = \angleBrackets{s_0, s_1, s_0', s_1'; \tilde \pi},
    \\[1mm]
    \begin{aligned}
    s_i^2
    = 1,&
    &&&&&&&
    {s_i'}^2
    &= 1,
    &&&&&&&
    &\tilde \pi^2 
    = 1,
    \\[1mm]
    \tilde \pi s_i
    = s_{i + 1} \tilde \pi,&
    &&&&&&&
    \tilde \pi s_i'
    &= s_{i + 1}' \tilde \pi,
    &&&&&&&
    &i 
    \in \quot{\mathbb{Z}}{2 \mathbb{Z}}.
    \end{aligned}
    \end{gathered}
\end{gather}
\end{minipage}
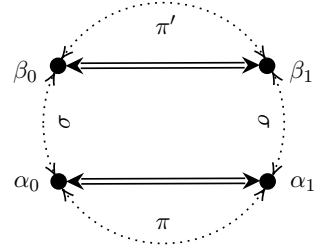
\begin{minipage}[l]{0.35\textwidth}
\begin{figure}[H]
    \vspace{-5mm}
    \centering
    \scalebox{0.85}{\tikzset{every picture/.style={line width=0.75pt}} 

\begin{tikzpicture}[x=0.75pt,y=0.75pt,yscale=-1,xscale=1]

\draw  [dash pattern={on 0.84pt off 2.51pt}] (221,135.5) .. controls (221,96.01) and (253.01,64) .. (292.5,64) .. controls (331.99,64) and (364,96.01) .. (364,135.5) .. controls (364,174.99) and (331.99,207) .. (292.5,207) .. controls (253.01,207) and (221,174.99) .. (221,135.5) -- cycle ;
\draw  [fill={rgb, 255:red, 0; green, 0; blue, 0 }  ,fill opacity=1 ] (225.73,170.16) .. controls (225.73,167.72) and (227.71,165.75) .. (230.14,165.75) .. controls (232.58,165.75) and (234.56,167.72) .. (234.56,170.16) .. controls (234.56,172.6) and (232.58,174.57) .. (230.14,174.57) .. controls (227.71,174.57) and (225.73,172.6) .. (225.73,170.16) -- cycle ;
\draw  [fill={rgb, 255:red, 0; green, 0; blue, 0 }  ,fill opacity=1 ] (349.73,170.16) .. controls (349.73,167.72) and (351.71,165.75) .. (354.14,165.75) .. controls (356.58,165.75) and (358.56,167.72) .. (358.56,170.16) .. controls (358.56,172.6) and (356.58,174.57) .. (354.14,174.57) .. controls (351.71,174.57) and (349.73,172.6) .. (349.73,170.16) -- cycle ;
\draw    (358,164.6) ;
\draw [shift={(357.44,166)}, rotate = 291.9] [color={rgb, 255:red, 0; green, 0; blue, 0 }  ][line width=0.75]    (6.56,-2.94) .. controls (4.17,-1.38) and (1.99,-0.4) .. (0,0) .. controls (1.99,0.4) and (4.17,1.38) .. (6.56,2.94)   ;
\draw    (351.19,176.38) -- (351.22,176.34) ;
\draw [shift={(352.38,174.71)}, rotate = 125.49] [color={rgb, 255:red, 0; green, 0; blue, 0 }  ][line width=0.75]    (6.56,-2.94) .. controls (4.17,-1.38) and (1.99,-0.4) .. (0,0) .. controls (1.99,0.4) and (4.17,1.38) .. (6.56,2.94)   ;
\draw    (243.55,168.63) -- (340.83,168.3)(243.56,171.63) -- (340.84,171.3) ;
\draw [shift={(349.83,169.77)}, rotate = 179.8] [fill={rgb, 255:red, 0; green, 0; blue, 0 }  ][line width=0.08]  [draw opacity=0] (10.72,-5.15) -- (0,0) -- (10.72,5.15) -- (7.12,0) -- cycle    ;
\draw [shift={(234.56,170.16)}, rotate = 359.8] [fill={rgb, 255:red, 0; green, 0; blue, 0 }  ][line width=0.08]  [draw opacity=0] (10.72,-5.15) -- (0,0) -- (10.72,5.15) -- (7.12,0) -- cycle    ;
\draw  [fill={rgb, 255:red, 0; green, 0; blue, 0 }  ,fill opacity=1 ] (225.4,101.49) .. controls (225.4,99.06) and (227.37,97.08) .. (229.81,97.08) .. controls (232.25,97.08) and (234.22,99.06) .. (234.22,101.49) .. controls (234.22,103.93) and (232.25,105.91) .. (229.81,105.91) .. controls (227.37,105.91) and (225.4,103.93) .. (225.4,101.49) -- cycle ;
\draw  [fill={rgb, 255:red, 0; green, 0; blue, 0 }  ,fill opacity=1 ] (349.4,101.49) .. controls (349.4,99.06) and (351.37,97.08) .. (353.81,97.08) .. controls (356.25,97.08) and (358.22,99.06) .. (358.22,101.49) .. controls (358.22,103.93) and (356.25,105.91) .. (353.81,105.91) .. controls (351.37,105.91) and (349.4,103.93) .. (349.4,101.49) -- cycle ;
\draw    (243.22,99.96) -- (340.49,99.63)(243.23,102.96) -- (340.51,102.63) ;
\draw [shift={(349.5,101.1)}, rotate = 179.8] [fill={rgb, 255:red, 0; green, 0; blue, 0 }  ][line width=0.08]  [draw opacity=0] (10.72,-5.15) -- (0,0) -- (10.72,5.15) -- (7.12,0) -- cycle    ;
\draw [shift={(234.22,101.49)}, rotate = 359.8] [fill={rgb, 255:red, 0; green, 0; blue, 0 }  ][line width=0.08]  [draw opacity=0] (10.72,-5.15) -- (0,0) -- (10.72,5.15) -- (7.12,0) -- cycle    ;
\draw    (358.2,107.52) -- (358.17,107.46) ;
\draw [shift={(357.2,105.72)}, rotate = 60.95] [color={rgb, 255:red, 0; green, 0; blue, 0 }  ][line width=0.75]    (6.56,-2.94) .. controls (4.17,-1.38) and (1.99,-0.4) .. (0,0) .. controls (1.99,0.4) and (4.17,1.38) .. (6.56,2.94)   ;
\draw    (226.76,107.88) -- (226.77,107.85) ;
\draw [shift={(227.66,106.06)}, rotate = 116.2] [color={rgb, 255:red, 0; green, 0; blue, 0 }  ][line width=0.75]    (6.56,-2.94) .. controls (4.17,-1.38) and (1.99,-0.4) .. (0,0) .. controls (1.99,0.4) and (4.17,1.38) .. (6.56,2.94)   ;
\draw    (227.15,164.71) ;
\draw [shift={(227.92,166.25)}, rotate = 243.43] [color={rgb, 255:red, 0; green, 0; blue, 0 }  ][line width=0.75]    (6.56,-2.94) .. controls (4.17,-1.38) and (1.99,-0.4) .. (0,0) .. controls (1.99,0.4) and (4.17,1.38) .. (6.56,2.94)   ;
\draw    (351.46,95.17) ;
\draw [shift={(352.38,96.56)}, rotate = 236.31] [color={rgb, 255:red, 0; green, 0; blue, 0 }  ][line width=0.75]    (6.56,-2.94) .. controls (4.17,-1.38) and (1.99,-0.4) .. (0,0) .. controls (1.99,0.4) and (4.17,1.38) .. (6.56,2.94)   ;
\draw    (233,95.72) ;
\draw [shift={(232.18,96.76)}, rotate = 308.11] [color={rgb, 255:red, 0; green, 0; blue, 0 }  ][line width=0.75]    (6.56,-2.94) .. controls (4.17,-1.38) and (1.99,-0.4) .. (0,0) .. controls (1.99,0.4) and (4.17,1.38) .. (6.56,2.94)   ;
\draw    (233.73,176.47) ;
\draw [shift={(232.64,175.02)}, rotate = 53.13] [color={rgb, 255:red, 0; green, 0; blue, 0 }  ][line width=0.75]    (6.56,-2.94) .. controls (4.17,-1.38) and (1.99,-0.4) .. (0,0) .. controls (1.99,0.4) and (4.17,1.38) .. (6.56,2.94)   ;

\draw (202.04,165.9) node [anchor=north west][inner sep=0.75pt]  [font=\large]  {$\alpha _{0}$};
\draw (366.04,165.88) node [anchor=north west][inner sep=0.75pt]  [font=\large]  {$\alpha _{1}$};
\draw (286.31,190.33) node [anchor=north west][inner sep=0.75pt]  [font=\large,rotate=-359.84]  {$\pi $};
\draw (201.71,97.23) node [anchor=north west][inner sep=0.75pt]  [font=\large]  {$\beta _{0}$};
\draw (365.71,97.22) node [anchor=north west][inner sep=0.75pt]  [font=\large]  {$\beta _{1}$};
\draw (285.98,69.66) node [anchor=north west][inner sep=0.75pt]  [font=\large,rotate=-359.84]  {$\pi '$};
\draw (356.6,128.91) node [anchor=north west][inner sep=0.75pt]  [font=\large,rotate=-90]  {$\sigma $};
\draw (229.06,140.25) node [anchor=north west][inner sep=0.75pt]  [font=\large,rotate=-270]  {$\sigma $};

\end{tikzpicture}}
\end{figure}
\end{minipage}

\begin{rem}
It is known that Bäcklund transformations of the third \Painleve equation form an affine Weyl group of type~$B_2^{(1)}$ \cite{okam4}. Indeed, by choosing another root variables and considering the reflections
\begin{align}
    &&
    \tilde \alpha_0 + 2 \tilde \alpha_1 + \tilde \alpha_2
    &= 1,
    &
    \tilde \alpha_1
    &:= \beta_1,
    &
    \tilde \alpha_2
    &:= \alpha_1 - \beta_1,
    &&&
    \tilde s_0
    &:= \sigma s_1' \pi' s_1' s_1 \pi s_1,
    &
    \tilde s_1
    &:= s_1',
    &
    \tilde s_2
    &:= \sigma,
    &&
\end{align}
one can obtain a group isomorphic to $W(B_2^{(1)})$:

\vspace{2mm}
\hspace{-5mm}
\begin{minipage}[l]{0.65\textwidth}
\begin{gather}
    \label{eq:WB2}
    \begin{gathered}
    {W}(B_2^{(1)})
    = \angleBrackets{\tilde s_0, \tilde s_1, \tilde s_2},
    \\[1mm]
    \begin{aligned}
    \tilde s_i^2
    &= 1,
    &&&&&
    (\tilde s_0 \, \tilde s_{2})^2
    &= 1,
    &&&&&
    (\tilde s_i \, \tilde s_{1})^4
    &= 1,
    &&&&&
    i 
    &= 0, 1, 2.
    \end{aligned}
    \end{gathered}
\end{gather}
\end{minipage}
\begin{minipage}[l]{0.35\textwidth}
\begin{figure}[H]
    \vspace{-5mm}
    \centering
    \scalebox{0.85}{\tikzset{every picture/.style={line width=0.75pt}} 

\begin{tikzpicture}[x=0.75pt,y=0.75pt,yscale=-1,xscale=1]

\draw  [fill={rgb, 255:red, 0; green, 0; blue, 0 }  ,fill opacity=1 ] (237.6,125.3) .. controls (237.6,127.73) and (235.62,129.71) .. (233.18,129.71) .. controls (230.75,129.71) and (228.77,127.73) .. (228.77,125.29) .. controls (228.77,122.86) and (230.75,120.88) .. (233.19,120.88) .. controls (235.62,120.88) and (237.6,122.86) .. (237.6,125.3) -- cycle ;
\draw  [fill={rgb, 255:red, 0; green, 0; blue, 0 }  ,fill opacity=1 ] (80.93,125.5) .. controls (80.93,127.93) and (78.95,129.91) .. (76.52,129.91) .. controls (74.08,129.91) and (72.11,127.93) .. (72.11,125.49) .. controls (72.11,123.06) and (74.08,121.08) .. (76.52,121.08) .. controls (78.96,121.08) and (80.93,123.06) .. (80.93,125.5) -- cycle ;
\draw [line width=0.75]    (228.27,127.02) -- (167.65,127.01)(228.27,124.02) -- (167.65,124.01) ;
\draw [shift={(158.65,125.51)}, rotate = 0.01] [fill={rgb, 255:red, 0; green, 0; blue, 0 }  ][line width=0.08]  [draw opacity=0] (10.72,-5.15) -- (0,0) -- (10.72,5.15) -- (7.12,0) -- cycle    ;
\draw  [fill={rgb, 255:red, 0; green, 0; blue, 0 }  ,fill opacity=1 ] (158.65,125.51) .. controls (158.65,127.94) and (156.67,129.92) .. (154.23,129.92) .. controls (151.8,129.92) and (149.82,127.94) .. (149.82,125.51) .. controls (149.82,123.07) and (151.8,121.09) .. (154.24,121.09) .. controls (156.67,121.09) and (158.65,123.07) .. (158.65,125.51) -- cycle ;
\draw [line width=0.75]    (140.82,127) -- (80.2,127)(140.82,124) -- (80.2,124) ;
\draw [shift={(149.82,125.51)}, rotate = 180.01] [fill={rgb, 255:red, 0; green, 0; blue, 0 }  ][line width=0.08]  [draw opacity=0] (10.72,-5.15) -- (0,0) -- (10.72,5.15) -- (7.12,0) -- cycle    ;

\draw (69.38,135.57) node [anchor=north west][inner sep=0.75pt]  [font=\large]  {$\tilde{\alpha }_{0}$};
\draw (227.38,135.55) node [anchor=north west][inner sep=0.75pt]  [font=\large]  {$\tilde{\alpha }_{2}$};
\draw (147.63,135.6) node [anchor=north west][inner sep=0.75pt]  [font=\large]  {$\tilde{\alpha }_{1}$};

\end{tikzpicture}}
\end{figure}
\end{minipage}
\end{rem}
\begin{rem}
The \ref{eq:P36sys} system may contain a non-abelian constant $h$ \cite{bobrova2022classification}
\begin{align}
    &
    \left\{
    \begin{array}{rcl}
         t \, \dot q
         &=& 2 q p q + \kappa_1 q + \kappa_2 q^2 + t \, h
         ,  
         \\[2mm]
         t \, \dot p
         &=& - 2 p q p - \kappa_1 p - \kappa_2 q p - \kappa_2 p q - \kappa_3
         .
    \end{array}
    \right.
\end{align}
It reduces to a system of the \ref{eq:P36sys} form by the map 
\begin{align}
    &&
    q
    &\mapsto \kappa_2 \, q,
    &
    p
    &\mapsto \kappa_2^{-1} p,
    &
    t
    &\mapsto \kappa_2 \kappa_4 \, t
    .
    &&
\end{align}
We set $\kappa_2 = 1$, since it is inessential. Then, the system admits the same Bäcklund transformations except of $s_0$, $s_0'$, $\pi$, and $\pi'$ generators. In particular, 
\begin{align}
    \pi \brackets{
    q, \, p, \, t; \, \alpha_0, \, \alpha_1, \, \beta_0, \, \beta_1
    }
    &= \brackets{
    - t q^{-1} \, h, \, 
    t^{-1} \, h^{-1} \, \brackets{
    q (p - 1) q
    + \beta_1 q
    }, \, 
    t; \, \alpha_1, \, 
    \alpha_0, \,
    \beta_0, \,
    \beta_1
    }
    ,
    \\[2mm]
    \pi' \brackets{
    q, \, p, \, t; \, \alpha_0, \, \alpha_1, \, \beta_0, \, \beta_1
    }
    &= \brackets{
    t q^{-1} \, h, \, 
    - t^{-1} \, h^{-1} \, \brackets{
    q p q
    + \alpha_1 q
    }, \, 
    t; \, 
    \alpha_0, \,
    \alpha_1, \,
    \beta_1, \, 
    \beta_0
    }
    .
\end{align}
\end{rem}

Turning to the discrete systems, let us determine two independent directions on the lattice $(\alpha_0, \alpha_1, \beta_0 \beta_1)$:
\begin{align}
    T
    &:= \pi' \sigma s_1 \sigma,
    &
    T (\alpha_0, \alpha_1, \beta_0, \beta_1)
    &= (\alpha_0, \alpha_1, \beta_0 + 1, \beta_1 - 1),
    \\
    T'
    &:= (\sigma s_1)^2 \pi' \pi
    ,
    &
    T' (\alpha_0, \alpha_1, \beta_0, \beta_1)
    &= (\alpha_0 - 1, \alpha_1 + 1, \beta_0 - 1, \beta_1 + 1).
\end{align}

\medskip
\subsubsection{{$T$-direction}}
\label{sec:dPD6_1}
Consider the translation operator $T = \pi' \sigma s_1 \sigma$, which acts on the parameters and variables by the rule
\begin{align}
    T_1 \brackets{
    q, \, p, \, t; \, \alpha_0, \, \alpha_1, \, \beta_0, \, \beta_1
    }
    &= \brackets{
    t q^{-1} + \bar \beta_1 (1 - \bar p)^{-1}, \,
    - t^{-1} (\alpha_1 q + q p q), \, 
    t; \,
    \alpha_0, \,
    \alpha_1, \,
    \beta_0 + 1, \,
    \beta_1 - 1
    }.
\end{align}
Introducing the variables $f:= q$ and $\bar g := t q^{-1} (1 - \bar p)$, it takes the form
\begin{align}
    \tag*{\ref{eq:dPD6}}
    &&
    \bar f \, f
    &= t + \bar \beta_1 t \bar g^{-1},
    &
    \bar g
    + g
    + [f \, \bar g, f^{-1}]
    &= \alpha_1 - \beta_1 + f + t f^{-1},
    &
    \bar \alpha_1
    &= \alpha_1,
    &
    \bar \beta_1
    &= \beta_1 - 1.
    &&
\end{align}
In the commutative case, this system corresponds to a discrete \Painleve equation related to the $D_6^{(1)}$-surface (see page 205 in \cite{sakai2001rational}).
\begin{proof}
Indeed, the $q$-dynamic is 
\begin{align}
    \bar q
    &= t q^{-1} + \bar \beta_1 (1 - \bar p)^{-1}
    = t q^{-1} + \bar \beta_1 \brackets{
    1 + t^{-1} (\alpha_1 q + q p q)
    }^{-1},
\end{align}
or, equivalently,
\begin{align}
    \bar q \, q
    &= t + \bar \beta_1 t \brackets{
    t q^{-1} + \alpha_1 + p q
    }^{-1}
    .
\end{align}
The latter expression suggests us the change of variables. Let us introduce $\bar{g} := t q^{-1} + \alpha_1 + p q \equiv t q^{-1} (1 - \bar p)$ and compute~$g$. Note that
\begin{align}
    &&
    \bar q
    &= t q^{-1} + \bar \beta_1 (1 - \bar p)^{-1}
    &
    \Leftrightarrow&
    &
    t \underline{q}^{-1}
    &= q - \beta_1 (1 - p)^{-1}.
    &&
\end{align}
Then
\begin{align}
    g
    &= t \underline{q}^{-1} (1 - p)
    = \brackets{
    q - \beta_1 (1 - p)^{-1}
    } (1 - p)
    = q - q p - \beta_1
    &
    \Rightarrow&
    &
    \bar{g} + g
    &= t q^{-1} + \alpha_1 
    + p q
    + q - q p - \beta_1.
\end{align}
By using the $p$-dynamic, one can compute 
$$p q - q p = t (\bar p q^{-1} - q^{-1} \bar p) = \bar g - q \, \bar g \, q^{-1} = [q^{-1}, q \, \bar g].$$ 
Therefore, we arrive at the wanted system.
\end{proof} 

The corresponding system of difference equations reads
\begin{align}
    &&
    &\left\{
    \begin{array}{rcl}
         f_{n} \, f_{n - 1}
         &=&  t + \beta_{1, n} t g_n^{-1}, 
         \\[2mm]
         g_{n + 1} + g_{n}
         + \left[ f_n \, g_{n + 1}, f_n^{-1} \right]
         &=& \alpha_{1, n} - \beta_{1, n}
         + f_n + t f_n^{-1}, 
    \end{array}
    \right.
    &
    \alpha_{1, n} 
    &= \alpha_{1},
    &
    \beta_{1, n}
    &= \beta_1 - n
    &&
\end{align}
and can be reduced to a non-commutative analog of the alt-d-P$_2$ equation:
\begin{align}
    \tag*{alt-d-P$_2$}
    \label{eq:altdP2}
    \begin{aligned}
    \beta_{1, n + 1} t (f_{n + 1} f_n - t)^{-1}
    &+ \beta_{1, n} t (f_{n} f_{n - 1} - t)^{-1}
    = \alpha_1 - \beta_{1, n}
    \\[1mm]
    &
    \begin{aligned}
    &+ \, \beta_{1, n + 1} t \LieBrackets{
    f_n^{-1}, f_n (f_{n + 1} f_{n} - t)^{-1}
    }
    + f_n + t f_n^{-1}
    ,
    &
    \beta_{1, n} 
    &= \beta_1 - n.
    \end{aligned}
    \end{aligned}
\end{align}
Its commutative version can be found, for instance, in Appendix A.4 of the paper \cite{van2007discrete}. 

Note that a continuous limit given by the transformation
\begin{gather}
\begin{aligned}
    f
    &= 1 + \varepsilon \, Q + \tfrac14 \, \varepsilon^2 \, P,
    &&&
    g
    &= - 1 + \varepsilon \, Q - \tfrac14 \, \varepsilon^2 \, P,
    &&&
    t
    &= - 1 - \tfrac12 \, \varepsilon^2 \, T,
\end{aligned}
\\[1mm]
\begin{aligned}
    \alpha_1
    &= - 2 + \tfrac12 (\rm{A}_1 - 1) \, \varepsilon^3,
    &&&
    \beta_1
    &= \tfrac14 (\mathrm{A}_1 - 1) \, \varepsilon^3
\end{aligned}
\end{gather}
brings the \ref{eq:dPD6} system to another non-abelian version of the P$_2$ equation obtained in \cite{Adler_Sokolov_2020_1}:
\begin{align}
    &
    \left\{
    \begin{array}{rcr}
         \dot Q
         &=& - Q^2 + P - \tfrac12 T,  
         \\[2mm]
         \dot P
         &=& 3 Q P - P Q + A_1.
    \end{array}
    \right.
\end{align}

\medskip
\subsubsection{{$T'$-direction}}
\label{sec:dPD6_2}
The operator $T' = (\sigma s_1)^2 \pi' \pi$ leads to the \ref{eq:dPD6'} dynamics
\begin{align}
    T' 
    &\brackets{
    q, \, p, \, t; \, \alpha_0, \, \alpha_1, \, \beta_0, \, \beta_1
    }
    \\
    &= \brackets{
    - q - \alpha_1 p^{-1} + \beta_1 (1 - p)^{-1}, \,
    1 - p - (1 + \alpha_1 + \beta_1) \bar q^{-1} - t \bar q^{-2}, \, 
    t; \,
    \alpha_0 - 1, \,
    \alpha_1 + 1, \,
    \beta_0 - 1, \, 
    \beta_1 + 1
    },
\end{align}
which is a non-abelian analog of the discrete dynamics given by eq. (8.240) in the paper \cite{kajiwara2017geometric}. 
The corresponding difference system reads
\begin{align}
    &&
    &\left\{
    \begin{array}{rcl}
         q_{n + 1} + q_{n}
         &=&  - \alpha_{1, n} p_n^{-1} 
         + \beta_{1, n} (1 - p_n)^{-1},
         \\[2mm]
         p_{n} + p_{n - 1}
         &=& 1 - (\alpha_{1, n} + \beta_{1, n} - 1) q_n^{-1}
         - t q_n^{-2}, 
    \end{array}
    \right.
    &
    \alpha_{1, n} 
    &= \alpha_{1} + n,
    &
    \beta_{1, n}
    &= \beta_1 + n
    .
    &&
\end{align}

\subsection{\texorpdfstring{$d$-P$(D_5)$}{dPD5}}
\label{sec:dPD5}
The corresponding abelian discrete equation is related to the P$_5$ system, whose non-abelian analog may be written as follows \cite{Kawakami_2015}
\begin{align}
    \tag*{P$_5$}
    \label{eq:P5sys}
    &
    \left\{
    \begin{array}{rcl}
         t \, \dot q
         &=& 2 q p q 
         - q p - p q - (\alpha_1 + \alpha_3) q
         + \alpha_1 
         + t q (q - 1)
         ,  
         \\[2mm]
         t \, \dot p
         &=& - 2 p q p + p^2 + (\alpha_1 + \alpha_3) p
         + t (p q + q p - p + \alpha_2)
         .
    \end{array}
    \right.
\end{align}
Here $p$ and $q$ are elements of $\mathcal{R}$. Set $\alpha_0 + \alpha_1 + \alpha_2 + \alpha_3 = 1$. Bäcklund transformations of the system are listed in the table below (cf. with \cite{bershtein2023hamiltonian} and Theorem \ref{thm:Aldsys}).
\begin{table}[H]
    \centering
    \begin{tabular}{c||cccc|ccc}
         & $\alpha_0$
         & $\alpha_1$
         & $\alpha_2$
         & $\alpha_3$
         & $q$
         & $p$
         & $t$
         \\[1mm]
         \hline\hline
         $s_0$
         & $- \alpha_0$
         & $\alpha_1 + \alpha_0$
         & $\alpha_2$
         & $\alpha_3 + \alpha_0$
         & $q + \alpha_0 (p + t)^{-1\phantom{\frac{1}{1}}}$
         & $p$
         & $t$
         \\[2mm]
         $s_1$
         & $\alpha_0 + \alpha_1$
         & $- \alpha_1$
         & $\alpha_2 + \alpha_1$
         & $\alpha_3$
         & $q$
         & $p - \alpha_1 q^{-1}$
         & $t$
         \\[2mm]
         $s_2$
         & $\alpha_0$
         & $\alpha_1 + \alpha_2$
         & $- \alpha_2$
         & $\alpha_3 + \alpha_2$
         & $q + \alpha_2 p^{-1}$
         & $p$
         & $t$
         \\[2mm]
         $s_3$
         & $\alpha_0 + \alpha_3$
         & $\alpha_1$
         & $\alpha_2 + \alpha_3$
         & $- \alpha_3$
         & $q$
         & $p - \alpha_3 (q - 1)^{-1}$
         & $t$
         \\[1mm]
         \hline
         $\pi$
         & $\alpha_1$
         & $\alpha_2$
         & $\alpha_3$
         & $\alpha_0$
         & $- t^{-1\phantom{\frac{1}{1}}} p$
         & $t q$
         & $- t$
         \\[2mm]
         $\pi'$
         & $\alpha_2$
         & $\alpha_1$
         & $\alpha_0$
         & $\alpha_3$
         & $q$
         & $p + t$
         & $- t$
    \end{tabular}
    \caption{BT for the \ref{eq:P5sys} system}
    \label{tab:BT_P5}
\end{table}
They form an extended affine Weyl group of type $A_3^{(1)}$
\begin{gather}
    \label{eq:WA3}
    \begin{gathered}
    \widetilde{W}(A_3^{(1)})
    = \angleBrackets{s_0, s_1, s_2, s_3; \pi},
    \\[1mm]
    \begin{aligned}
    (s_i \, s_{j})^2
    &= 1 \,\,\, (j \neq i \pm 1),
    &&&&&
    (s_i \, s_{i + 1})^3
    &= 1,
    &&&&&
    \pi^4
    &= 1,
    &&&&&
    \pi s_i
    &= s_{i + 1} \pi,
    &&&&&
    i, j
    &\in \quot{\mathbb{Z}}{4 \mathbb{Z}}.
    \end{aligned}
    \end{gathered}
\end{gather}
The corresponding Cartan matrix and Dynkin diagram are given by
\begin{table}[H]
    \vspace{-5mm}
    \begin{tabular}{ccc}
    $C
    = 
    \begin{pmatrix}
    2 & - 1 & 0 & - 1 \\[0.9mm] 
    - 1 & 2 & - 1 & 0 \\[0.9mm]
    0 & - 1 & 2 & - 1 \\[0.9mm]
    - 1 & 0 & - 1 & 2
    \end{pmatrix}
    $
    & \hspace{2cm} &
    \begin{tabular}{c}
    \\[-1mm]
    \scalebox{0.85}{\tikzset{every picture/.style={line width=0.75pt}} 

\begin{tikzpicture}[x=0.75pt,y=0.75pt,yscale=-1,xscale=1]

\draw  [dash pattern={on 0.84pt off 2.51pt}] (57.5,110.17) .. controls (57.5,70.68) and (89.51,38.67) .. (129,38.67) .. controls (168.49,38.67) and (200.5,70.68) .. (200.5,110.17) .. controls (200.5,149.66) and (168.49,181.67) .. (129,181.67) .. controls (89.51,181.67) and (57.5,149.66) .. (57.5,110.17) -- cycle ;
\draw  [fill={rgb, 255:red, 0; green, 0; blue, 0 }  ,fill opacity=1 ] (175.02,59.73) .. controls (175.02,57.29) and (177,55.32) .. (179.44,55.32) .. controls (181.87,55.32) and (183.85,57.29) .. (183.85,59.73) .. controls (183.85,62.17) and (181.87,64.14) .. (179.44,64.14) .. controls (177,64.14) and (175.02,62.17) .. (175.02,59.73) -- cycle ;
\draw  [fill={rgb, 255:red, 0; green, 0; blue, 0 }  ,fill opacity=1 ] (74.15,160.6) .. controls (74.15,158.17) and (76.13,156.19) .. (78.56,156.19) .. controls (81,156.19) and (82.98,158.17) .. (82.98,160.6) .. controls (82.98,163.04) and (81,165.02) .. (78.56,165.02) .. controls (76.13,165.02) and (74.15,163.04) .. (74.15,160.6) -- cycle ;
\draw  [fill={rgb, 255:red, 0; green, 0; blue, 0 }  ,fill opacity=1 ] (175.02,160.6) .. controls (175.02,158.17) and (177,156.19) .. (179.44,156.19) .. controls (181.87,156.19) and (183.85,158.17) .. (183.85,160.6) .. controls (183.85,163.04) and (181.87,165.02) .. (179.44,165.02) .. controls (177,165.02) and (175.02,163.04) .. (175.02,160.6) -- cycle ;
\draw    (183.83,155.43) ;
\draw [shift={(183.05,156.39)}, rotate = 309.44] [color={rgb, 255:red, 0; green, 0; blue, 0 }  ][line width=0.75]    (6.56,-2.94) .. controls (4.17,-1.38) and (1.99,-0.4) .. (0,0) .. controls (1.99,0.4) and (4.17,1.38) .. (6.56,2.94)   ;
\draw    (84,165.75) -- (83.8,165.59) ;
\draw [shift={(82.24,164.34)}, rotate = 38.8] [color={rgb, 255:red, 0; green, 0; blue, 0 }  ][line width=0.75]    (6.56,-2.94) .. controls (4.17,-1.38) and (1.99,-0.4) .. (0,0) .. controls (1.99,0.4) and (4.17,1.38) .. (6.56,2.94)   ;
\draw    (73.89,64.39) ;
\draw [shift={(74.91,63.21)}, rotate = 130.53] [color={rgb, 255:red, 0; green, 0; blue, 0 }  ][line width=0.75]    (6.56,-2.94) .. controls (4.17,-1.38) and (1.99,-0.4) .. (0,0) .. controls (1.99,0.4) and (4.17,1.38) .. (6.56,2.94)   ;
\draw   (78.56,59.73) -- (179.44,59.73) -- (179.44,160.6) -- (78.56,160.6) -- cycle ;
\draw  [fill={rgb, 255:red, 0; green, 0; blue, 0 }  ,fill opacity=1 ] (74.15,59.73) .. controls (74.15,57.29) and (76.13,55.32) .. (78.56,55.32) .. controls (81,55.32) and (82.98,57.29) .. (82.98,59.73) .. controls (82.98,62.17) and (81,64.14) .. (78.56,64.14) .. controls (76.13,64.14) and (74.15,62.17) .. (74.15,59.73) -- cycle ;
\draw    (174.45,55.27) ;
\draw [shift={(175.43,56.21)}, rotate = 223.79] [color={rgb, 255:red, 0; green, 0; blue, 0 }  ][line width=0.75]    (6.56,-2.94) .. controls (4.17,-1.38) and (1.99,-0.4) .. (0,0) .. controls (1.99,0.4) and (4.17,1.38) .. (6.56,2.94)   ;
\draw    (78.56,59.73) -- (78.56,156.19) ;
\draw    (179.44,64.14) -- (179.44,160.6) ;
\draw    (82.98,59.73) -- (179.44,59.73) ;
\draw    (78.56,160.6) -- (175.02,160.6) ;
\draw  [dash pattern={on 0.84pt off 2.51pt}]  (78.56,59.73) -- (179.44,160.6) ;
\draw    (84.17,65.29) -- (83.97,65.08) ;
\draw [shift={(82.55,63.67)}, rotate = 45] [color={rgb, 255:red, 0; green, 0; blue, 0 }  ][line width=0.75]    (6.56,-2.94) .. controls (4.17,-1.38) and (1.99,-0.4) .. (0,0) .. controls (1.99,0.4) and (4.17,1.38) .. (6.56,2.94)   ;
\draw    (174.44,155.82) ;
\draw [shift={(175.5,156.88)}, rotate = 225] [color={rgb, 255:red, 0; green, 0; blue, 0 }  ][line width=0.75]    (6.56,-2.94) .. controls (4.17,-1.38) and (1.99,-0.4) .. (0,0) .. controls (1.99,0.4) and (4.17,1.38) .. (6.56,2.94)   ;

\draw (46.26,156) node [anchor=north west][inner sep=0.75pt]  [font=\large]  {$\alpha _{1}$};
\draw (193.66,155.97) node [anchor=north west][inner sep=0.75pt]  [font=\large]  {$\alpha _{2}$};
\draw (46.27,55.58) node [anchor=north west][inner sep=0.75pt]  [font=\large,rotate=-0.45]  {$\alpha _{0}$};
\draw (122.81,168) node [anchor=north west][inner sep=0.75pt]  [font=\large,rotate=-359.84]  {$\pi $};
\draw (193.67,55.58) node [anchor=north west][inner sep=0.75pt]  [font=\large,rotate=-0.45]  {$\alpha _{3}$};
\draw (122.81,43.79) node [anchor=north west][inner sep=0.75pt]  [font=\large,rotate=-359.84]  {$\pi $};
\draw (62.21,113.91) node [anchor=north west][inner sep=0.75pt]  [font=\large,rotate=-270.91]  {$\pi $};
\draw (195.76,102.42) node [anchor=north west][inner sep=0.75pt]  [font=\large,rotate=-89.78]  {$\pi $};
\draw (137.5,92.5) node [anchor=north west][inner sep=0.75pt]  [font=\large,rotate=-45]  {$\pi '$};

\end{tikzpicture}}
    \end{tabular}
    \end{tabular}
    \vspace{-5mm}
\end{table}

Turning to the dyscrete dynamics, we consider two translations $T := (\pi_1 \pi_2)^2 s_0 s_2 s_1 s_3$ and $T' := s_3 s_0 s_1 \pi_2 \pi_1$ that act on the parameters~as
\begin{align}
    T \brackets{
    \alpha_0, \, \alpha_1, \, \alpha_2, \, \alpha_3
    }
    &= \brackets{
    \alpha_0 + 1, \, \alpha_1 - 1, \, \alpha_2 + 1, \, \alpha_3 - 1
    },
    \\[1mm]
    T' \brackets{
    \alpha_0, \, \alpha_1, \, \alpha_2, \, \alpha_3
    }
    &= \brackets{
    \alpha_0, \, \alpha_1, \, \alpha_2 - 1, \, \alpha_3 + 1
    }.
\end{align}

\medskip
\subsubsection{{$T$-direction}}
\label{sec:dPD5_1}
In this simple case we obtain the following discrete $q$, $p$ dynamics
\begin{align}
    \tag*{\ref{eq:dPD5}}
    \bar q + q
    &= 1 - \alpha_2 p^{-1} - \alpha_0 (p + t)^{-1},
    &
    \bar p + p
    &= - t + (\alpha_1 - 1) \bar q^{-1} + (\alpha_3 - 1) (\bar q - 1)^{-1},
\end{align}
which is a non-abelian version of the d-P$_3$ system from \cite{sakai2001rational} (see page 205). The corresponding difference system is
\begin{gather}
    \tag*{d-P$_3$}
    \label{eq:dP3}
    \left\{
    \begin{array}{rcl}
         q_{n + 1} + q_{n}
         &=& 1 - \alpha_{2, n} p_n^{-1}
         - \alpha_{0, n} (p_n + t)^{-1},
         \\[2mm]
         p_{n} + p_{n - 1}
         &=& - t 
         + \alpha_{1, n} q_n^{-1}
         + \alpha_{3, n} (q_n - 1)^{-1}
         , 
    \end{array}
    \right.
    \\[2mm]
    \begin{aligned}
    \alpha_{0, n} 
    &= \alpha_{0} + n,
    &&&&&
    \alpha_{1, n} 
    &= \alpha_{1} - n,
    &&&&&
    \alpha_{2, n}
    &= \alpha_2 + n,
    &&&&&
    \alpha_{3, n}
    &= \alpha_3 - n
    .
    \end{aligned}
\end{gather}

\medskip
\subsubsection{{$T'$-direction}}
\label{sec:dPD5_2}
The $T'$-operator, $T' = s_3 s_0 s_1 \pi_2 \pi_1$, gives the discrete equations
\begin{align}
    t \, \bar q
    &= t + p - \alpha_3 (q - 1)^{-1} + (1 - \alpha_2) t^{-1} \bar p^{-1},
    &
    t^{-1} \bar p
    &= - q - (\alpha_0 + \alpha_3) \brackets{
    t + p - \alpha_3 (q - 1)^{-1}
    }^{-1},
\end{align}
that can be rewritten in the form 
\begin{align}
    \tag*{\ref{eq:dPD5'}}
    \bar f + f + \LieBrackets{
    g^{-1}, \bar f \, g
    }
    &= - (\alpha_0 + \alpha_2) - t g - \alpha_2 (g - 1)^{-1},
    &
    g \, \bar g
    &= - t^{-1} \bar f (\bar f + \alpha_0) (\bar f - \alpha_3)^{-1}
    ,
\end{align}
where $g := t^{-1} p + 1$ and $\bar f := (p + t) (q - 1)$.
The latter is a non-abelian analog of the d-P$_4$ equation (see page 204 in \cite{sakai2001rational} or eq. (8.236) in \cite{kajiwara2017geometric}). 
\begin{proof}
Indeed, first of all let us make the following observation
\begin{align}
    &&
    t^{-1} \bar p + 1
    &= 1 - q - (\alpha_0 + \alpha_3)
    (q - 1) \brackets{
    (p + t) (q - 1) - \alpha_3
    }^{-1}
    \\[2mm]
    \Leftrightarrow&
    &
    t^{-1} \bar p + 1
    &= - (q - 1) \brackets{
    (p + t) (q - 1) + \alpha_0
    } \brackets{
    (p + t) (q - 1) - \alpha_3
    }^{-1}
    .
    &&&&
\end{align}
This equation suggests us a transformation as above. In order to find the $f$-dynamic, one needs to compute $f = (\underline{p} + t) (\underline{q} - 1)$. Thanks to the $T'$-operator, it is easy to obtain
\begin{align}
    t \underline{q}
    &= - p + (\alpha_1 + \alpha_2) (q + \alpha_2 p^{-1})^{-1}
    ,
    &
    \underline{p} + t
    &= t (q + \alpha_2 p^{-1})
    - (1 - \alpha_3) (\underline{q} - 1)^{-1}
    ,
\end{align}
and then
\begin{align}
    (\underline{p} + t) (\underline{q} - 1)
    &= t (q + \alpha_2 p^{-1}) (\underline{q} - 1)
    - (1 - \alpha_3)
    \\
    &= (q + \alpha_2 p^{-1}) \brackets{
    - p + (\alpha_1 + \alpha_2) (q + \alpha_2 p^{-1})^{-1}
    - t
    }
    - (1 - \alpha_3)
    \\
    &= - q (p + t)
    - \alpha_2 p^{-1} (p + t)
    - \alpha_0
    \\
    &= - (q - 1) (p + t) - (p + t)
    - \alpha_2 - \alpha_2 (t^{-1} p + 1 - 1)^{-1} - \alpha_0
    \\
    &= - (q - 1) (p + t) 
    - t \, g - \alpha_2 (g - 1)^{-1}
    - (\alpha_0 + \alpha_2).
\end{align}
Since
\begin{align}
    (q - 1) (p + t) - (p + t) (q - 1)
    &= (p + t)^{-1} (p + t) (q - 1) (p + t)
    - (p + t) (q - 1) (p + t) (p + t)^{-1} 
    \\
    &= \LieBrackets{
    t^{-1} ( t^{-1} p + 1)^{-1}, 
    (p + t) (q - 1) \, \, t (t^{-1} p + 1)
    }
    = \LieBrackets{
    g^{-1}, \bar f \, g
    },
\end{align}
we arrive at the desired result. 
\end{proof}
The \ref{eq:dPD5'} dynamics can be represented as the system of difference equations. Namely, we have
\begin{gather}
    \tag*{d-P$_4$}
    \label{eq:dP4}
    \left\{
    \begin{array}{rcl}
         f_{n + 1} + f_{n} + \LieBrackets{
         g_n^{-1}, f_{n + 1} \, g_n
         }
         &=& - (\alpha_{0, n} + \alpha_{2, n})
         - t \, g_n - \alpha_{2, n} (g_n - 1)^{-1}
         ,
         \\[2mm]
         g_{n - 1} \, g_{n}
         &=& - t^{-1} 
         f_n (f_n + \alpha_{0, n}) (f_n - \alpha_{3, n} + 1)
         , 
    \end{array}
    \right.
    \\[2mm]
    \begin{aligned}
    \alpha_{0, n} 
    &= \alpha_{0},
    &&&&&
    \alpha_{1, n} 
    &= \alpha_{1},
    &&&&&
    \alpha_{2, n}
    &= \alpha_2 - n,
    &&&&&
    \alpha_{3, n}
    &= \alpha_3 + n
    .
    \end{aligned}
\end{gather}

\subsection{\texorpdfstring{$d$-P$(D_4)$}{dPD4}}
\label{sec:dPD4}
The last case we are going to consider is connected with the P$_6$ system. Its non-abelian analog may be presented as the system \cite{bobrova2022classification}
\begin{align}
    \tag*{P$_6$}
    \label{eq:P6sys}
    &
    \left\{
    \begin{array}{rcl}
         t (t - 1) \, \dot q
         &=& q^2 p q + q p q^2 - 2 q p q 
         + \beta [q, [q, p]]
         + (\alpha_1 + 2 \alpha_2) q^2
         - (\alpha_1 + 2 \alpha_2 + \alpha_3) q
         \\[2mm]
         &&
         \hspace{3.4cm}
         +\, t \, \brackets{
         - 2 q p q + \gamma [q, [q, p]] + q p + p q + (\alpha_3 + \alpha_4) q - \alpha_4
         }
         ,  
         \\[3mm]
         t (t - 1) \, \dot p
         &=& - q p q p - p q^2 p - p q p q
         + 2 p q p + \beta [p, [q, p]]
         - (\alpha_1 + 2 \alpha_2) (q p + p q)
         \\[2mm]
         &&
         + \, (\alpha_1 + 2 \alpha_2 + \alpha_3) p
         - \alpha_2 (\alpha_1 + \alpha_2)
         + t \brackets{
         2 p q p + \gamma [p, [q, p]]
         - p^2 - (\alpha_3 + \alpha_4) p
         }
         .
    \end{array}
    \right.
\end{align}
Here $\beta$, $\gamma$ are additional arbitrary abelian parameters which we will fix below. Let~$\alpha_0 + \alpha_1 + 2 \alpha_2 + \alpha_3 + \alpha_4 = 1$. When $\beta = \gamma = - \tfrac13$, the system has the following Bäcklund transformations (cf. with \cite{bobrova2022non} and \cite{bershtein2023hamiltonian})
\begin{table}[H]
    \centering
    \begin{tabular}{c||ccccc|ccc}
         & $\alpha_0$
         & $\alpha_1$
         & $\alpha_2$
         & $\alpha_3$
         & $\alpha_4$
         & $q$
         & $p$
         & $t$
         \\[1mm]
         \hline\hline
         $s_0$
         & $- \alpha_0$
         & $\alpha_1$
         & $\alpha_2 + \alpha_0$
         & $\alpha_3$
         & $\alpha_4$
         & $q$
         & $\,\, \,p - \alpha_0 (q - t)^{-1\phantom{\frac{1}{}}}$
         & $t$
         \\[2mm]
         $s_1$
         & $\alpha_0$
         & $- \alpha_1$
         & $\alpha_2 + \alpha_1$
         & $\alpha_3$
         & $\alpha_4$
         & $q$
         & $p$
         & $t$
         \\[2mm]
         $s_2$
         & $\alpha_0 + \alpha_2$
         & $\alpha_1 + \alpha_2$
         & $- \alpha_2$
         & $\alpha_3 + \alpha_2$
         & $\alpha_4 + \alpha_2$
         & $q + \alpha_2 p^{-1}$
         & $p$
         & $t$
         \\[2mm]
         $s_3$
         & $\alpha_0$
         & $\alpha_1$
         & $\alpha_2 + \alpha_3$
         & $- \alpha_3$
         & $\alpha_4$
         & $q$
         & $p - \alpha_3 (q - 1)^{-1}$
         & $t$
         \\[2mm]
         $s_4$
         & $\alpha_0$
         & $\alpha_1$
         & $\alpha_2 + \alpha_4$
         & $\alpha_3$
         & $- \alpha_4$
         & $q$
         & $p - \alpha_4 q^{-1}$
         & $t$
         \\[1mm]
         \hline
         $r_1$
         & $\alpha_0$
         & $\alpha_1$
         & $\alpha_2$
         & $\alpha_4$
         & $\alpha_3$
         & $1 - q$
         & $- p$
         & $1 - t$
         \\[2mm]
         $r_2$
         & $\alpha_3$
         & $\alpha_1$
         & $\alpha_2$
         & $\alpha_0$
         & $\alpha_4$
         & $t^{-1} q$
         & $t p$
         & $t^{-1}$
         \\[2mm]
         $r_3$
         & $\alpha_0$
         & $\alpha_4$
         & $\alpha_2$
         & $\alpha_3$
         & $\alpha_1$
         & $q^{-1}$
         & $- q p q - \alpha_2 q$
         & $t^{-1}$
    \end{tabular}
    \caption{BT for the \ref{eq:P6sys} system with $\beta = \gamma = - \frac13$}
    \label{tab:BT_P6}
\end{table}
\begin{rem}
Unlike the paper \cite{bershtein2023hamiltonian}, in the case of the permutations $r_1$, $r_2$, $r_3$, we do not need to conjugate the variables $q$ and $p$ by the element $g = t^{[p, q]}$, thanks to the additional parameters $\beta$, $\gamma$ in the~\ref{eq:P6sys}~system.
\end{rem}

The permutations $r_1$, $r_2$, $r_3$ form a group isomorphic to $S_4$:
\begin{gather}
    \label{eq:S4}
    \begin{gathered}
    S_4
    = \angleBrackets{
    r_1, r_2, r_3
    },
    \\[1mm]
    \begin{aligned}
        r_k^2
        &= 1,
        &&&&&
        (r_2 r_3)^2
        &= 1,
        &&&&&
        (r_1 r_2)^3
        &= 1,
        &&&&&
        (r_1 r_3)^3
        &= 1,
        &&&&&
        k
        &= 1, 2, 3.
    \end{aligned}
    \end{gathered}
\end{gather}
Set $\pi_1 := r_2 r_3$, $\pi_2 := r_1 r_2 r_3 r_1$, $\pi_3 := (r_1 r_2 r_3)^2$. Taking them together with the reflections, we obtain the extended affine Wey group of type $D_4^{(1)}$ 
\begin{gather}
    \label{eq:WD4}
    \begin{gathered}
    \widetilde{W}(D_4^{(1)})
    = \angleBrackets{s_0, s_1, s_2, s_3, s_4; \pi_1, \pi_2, \pi_3},
    \\[1mm]
    \begin{aligned}
    s_i^2
    &= 1,
    &&&&&
    (s_i \, s_{j})^2
    &= 1 \,\,\, (j \neq i \neq 2),
    &&&&&
    (s_i \, s_{2})^3
    &= 1,
    &&&&&
    i, j
    &= 0, 1, 2, 3, 4,
    \end{aligned}
    \\[1mm]
    \begin{aligned}
    \pi_k^2
    &= 1,
    &&&&&
    (\pi_2 \, \pi_3)^2
    &= 1,
    &&&&&
    \pi_k s_i
    &= s_{\sigma_k(i)} \pi_k,
    &&&&&
    k
    &= 1, 2, 3.
    \end{aligned}
    \end{gathered}
\end{gather}
The Cartan matrix and Dynkin diagram are presented below. By these data, one can easily recover the~$\widetilde{W}$-action on the root variables as well as the fundamental relations for its generators.
\begin{table}[H]
    \vspace{-3mm}
    \begin{tabular}{ccc}
    $C
    = 
    \begin{pmatrix}
    2 & 0 & - 1 & 0 & 0 \\[0.9mm] 
    0 & 2 & -1 & 0 & 0 \\[0.9mm]
    - 1 & -1 & 2 & -1 & -1 \\[0.9mm]
    0 & 0 & -1 & 2 & 0 \\[0.9mm]
    0 & 0 & -1 & 0 & 2
    \end{pmatrix}
    $
    & \hspace{2cm} &
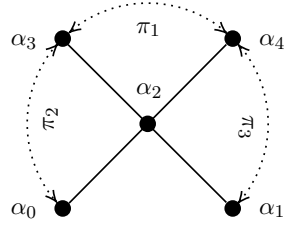
    \begin{tabular}{c}
    \\[-5mm]
    \scalebox{0.85}{\tikzset{every picture/.style={line width=0.75pt}} 

\begin{tikzpicture}[x=0.75pt,y=0.75pt,yscale=-1,xscale=1]

\draw  [draw opacity=0][dash pattern={on 0.84pt off 2.51pt}] (78.9,161.2) .. controls (58.06,140.74) and (51.12,108.67) .. (63.92,80.55) .. controls (80.28,44.61) and (122.68,28.73) .. (158.62,45.09) .. controls (194.56,61.45) and (210.43,103.84) .. (194.08,139.78) .. controls (190.52,147.6) and (185.73,154.47) .. (180.04,160.27) -- (129,110.17) -- cycle ; \draw  [dash pattern={on 0.84pt off 2.51pt}] (78.9,161.2) .. controls (58.06,140.74) and (51.12,108.67) .. (63.92,80.55) .. controls (80.28,44.61) and (122.68,28.73) .. (158.62,45.09) .. controls (194.56,61.45) and (210.43,103.84) .. (194.08,139.78) .. controls (190.52,147.6) and (185.73,154.47) .. (180.04,160.27) ;  
\draw    (184.36,155.45) -- (184.19,155.65) ;
\draw [shift={(182.93,157.2)}, rotate = 309.09] [color={rgb, 255:red, 0; green, 0; blue, 0 }  ][line width=0.75]    (6.56,-2.94) .. controls (4.17,-1.38) and (1.99,-0.4) .. (0,0) .. controls (1.99,0.4) and (4.17,1.38) .. (6.56,2.94)   ;
\draw    (73.27,154.71) -- (73.51,155.01) ;
\draw [shift={(74.74,156.58)}, rotate = 231.7] [color={rgb, 255:red, 0; green, 0; blue, 0 }  ][line width=0.75]    (6.56,-2.94) .. controls (4.17,-1.38) and (1.99,-0.4) .. (0,0) .. controls (1.99,0.4) and (4.17,1.38) .. (6.56,2.94)   ;
\draw    (74.34,64.43) ;
\draw [shift={(75.28,63.46)}, rotate = 134.08] [color={rgb, 255:red, 0; green, 0; blue, 0 }  ][line width=0.75]    (6.56,-2.94) .. controls (4.17,-1.38) and (1.99,-0.4) .. (0,0) .. controls (1.99,0.4) and (4.17,1.38) .. (6.56,2.94)   ;

\draw    (183.68,64.47) ;
\draw [shift={(182.89,63.51)}, rotate = 50.39] [color={rgb, 255:red, 0; green, 0; blue, 0 }  ][line width=0.75]    (6.56,-2.94) .. controls (4.17,-1.38) and (1.99,-0.4) .. (0,0) .. controls (1.99,0.4) and (4.17,1.38) .. (6.56,2.94)   ;

\draw    (174.54,55.56) ;
\draw [shift={(175.49,56.35)}, rotate = 219.73] [color={rgb, 255:red, 0; green, 0; blue, 0 }  ][line width=0.75]    (6.56,-2.94) .. controls (4.17,-1.38) and (1.99,-0.4) .. (0,0) .. controls (1.99,0.4) and (4.17,1.38) .. (6.56,2.94)   ;

\draw    (83.54,55.31) ;
\draw [shift={(82.65,56.04)}, rotate = 320.84] [color={rgb, 255:red, 0; green, 0; blue, 0 }  ][line width=0.75]    (6.56,-2.94) .. controls (4.17,-1.38) and (1.99,-0.4) .. (0,0) .. controls (1.99,0.4) and (4.17,1.38) .. (6.56,2.94)   ;

\draw  [fill={rgb, 255:red, 0; green, 0; blue, 0 }  ,fill opacity=1 ] (175.02,59.73) .. controls (175.02,57.29) and (177,55.32) .. (179.44,55.32) .. controls (181.87,55.32) and (183.85,57.29) .. (183.85,59.73) .. controls (183.85,62.17) and (181.87,64.14) .. (179.44,64.14) .. controls (177,64.14) and (175.02,62.17) .. (175.02,59.73) -- cycle ;
\draw  [fill={rgb, 255:red, 0; green, 0; blue, 0 }  ,fill opacity=1 ] (74.15,160.6) .. controls (74.15,158.17) and (76.13,156.19) .. (78.56,156.19) .. controls (81,156.19) and (82.98,158.17) .. (82.98,160.6) .. controls (82.98,163.04) and (81,165.02) .. (78.56,165.02) .. controls (76.13,165.02) and (74.15,163.04) .. (74.15,160.6) -- cycle ;
\draw  [fill={rgb, 255:red, 0; green, 0; blue, 0 }  ,fill opacity=1 ] (175.02,160.6) .. controls (175.02,158.17) and (177,156.19) .. (179.44,156.19) .. controls (181.87,156.19) and (183.85,158.17) .. (183.85,160.6) .. controls (183.85,163.04) and (181.87,165.02) .. (179.44,165.02) .. controls (177,165.02) and (175.02,163.04) .. (175.02,160.6) -- cycle ;
\draw  [fill={rgb, 255:red, 0; green, 0; blue, 0 }  ,fill opacity=1 ] (74.15,59.73) .. controls (74.15,57.29) and (76.13,55.32) .. (78.56,55.32) .. controls (81,55.32) and (82.98,57.29) .. (82.98,59.73) .. controls (82.98,62.17) and (81,64.14) .. (78.56,64.14) .. controls (76.13,64.14) and (74.15,62.17) .. (74.15,59.73) -- cycle ;
\draw    (175.67,63.24) -- (82.33,157.46) ;
\draw    (81.67,63.24) -- (175.67,157.02) ;
\draw  [fill={rgb, 255:red, 0; green, 0; blue, 0 }  ,fill opacity=1 ] (124.59,110.35) .. controls (124.59,107.91) and (126.56,105.94) .. (129,105.94) .. controls (131.44,105.94) and (133.41,107.91) .. (133.41,110.35) .. controls (133.41,112.79) and (131.44,114.76) .. (129,114.76) .. controls (126.56,114.76) and (124.59,112.79) .. (124.59,110.35) -- cycle ;

\draw (46.26,156) node [anchor=north west][inner sep=0.75pt]  [font=\large]  {$\alpha _{0}$};
\draw (193.66,155.97) node [anchor=north west][inner sep=0.75pt]  [font=\large]  {$\alpha _{1}$};
\draw (46.27,55.58) node [anchor=north west][inner sep=0.75pt]  [font=\large,rotate=-0.45]  {$\alpha _{3}$};
\draw (193.67,55.58) node [anchor=north west][inner sep=0.75pt]  [font=\large,rotate=-0.45]  {$\alpha _{4}$};
\draw (120.6,85.08) node [anchor=north west][inner sep=0.75pt]  [font=\large]  {$\alpha _{2}$};
\draw (193.03,103.65) node [anchor=north west][inner sep=0.75pt]  [font=\large,rotate=-90]  {$\pi _{3}$};
\draw (65.61,116.01) node [anchor=north west][inner sep=0.75pt]  [font=\large,rotate=-270.13]  {$\pi _{2}$};
\draw (120.61,48.38) node [anchor=north west][inner sep=0.75pt]  [font=\large,rotate=-0.07]  {$\pi _{1}$};

\end{tikzpicture}}
    \end{tabular}
    \end{tabular}
    \vspace{-5mm}
\end{table}

Proceeding to the discrete dynamic, let us consider the translation operator $T := s_2 s_4 s_1 s_2 s_0 s_3 \pi_2 \pi_3$ that acts on the parameters as
\begin{align}
    T \brackets{
    \alpha_0, \,
    \alpha_1, \,
    \alpha_2, \,
    \alpha_3, \, 
    \alpha_4
    }
    &= \brackets{
    \alpha_0 - 1, \,
    \alpha_1, \,
    \alpha_2 + 1, \,
    \alpha_3 - 1, \, 
    \alpha_4
    }.
\end{align}
The corresponding $q$, $p$ dynamic is
\begin{align}
    &\bar q
    = t \brackets{
    (q + \alpha_2 p^{-1})
    + (\alpha_1 + \alpha_2 + \alpha_4) \brackets{
    p - (\alpha_2 + \alpha_4) (q + \alpha_2 p^{-1})^{-1}
    }^{-1}
    }^{-1},
    \\[2mm]
    &
    \begin{aligned}
    \bar p
    = - (1 + \alpha_2) \bar q^{-1}
    - t \, \bar q^{-1} \brackets{
    p - (\alpha_2 + \alpha_4) (q + \alpha_2 p^{-1})^{-1}
    } \bar q^{-1}
    &+ (\alpha_0 - 1) \, t \, \bar q^{-1} (1 - t \bar q^{-1})^{-1} \bar q^{-1}
    \\[1mm]
    &+ \, (\alpha_3 - 1) \, t \, \bar q^{-1} (t - t \bar q^{-1})^{-1} \bar q^{-1},
    \end{aligned}
\end{align}
or, in terms of $f := q$ and $g := q p$, 
\begin{gather}
    \tag*{\ref{eq:dPD4}}
    \begin{aligned}
        f \, \bar f
        &= t \, g \, (g + \alpha_2)^{-1} \,
        (g - \alpha_4) \, (g + \alpha_1 + \alpha_2)^{-1}
        ,
        \\[1mm]
        \bar g + g 
        + \LieBrackets{f^{-1}, \, g \, f}
        &= (\alpha_0 + \alpha_3 + \alpha_4 - 2)
        + (\alpha_3 - 1) \, (\bar f - 1)^{-1}
        + (\alpha_0 - 1) \, t \, (\bar f - t)^{-1}
        .
    \end{aligned}
\end{gather}
Under the commutative reduction, this dynamic coincides with a discrete \Painleve equation related to the~$D_4^{(1)}$-surface (see page 204 in \cite{sakai2001rational} or eq. (8.227) in \cite{kajiwara2017geometric}) and is also known as the d-P$_5$ equation.
\begin{proof}
Indeed, the $q$-dynamic may be rewritten as follows
\begin{align}
    t^{-1} \bar q
    &= \brackets{
    p - (\alpha_2 + \alpha_4) (q + \alpha_2 p^{-1})^{-1}
    } \,
    \brackets{
    q p + \alpha_1 + \alpha_2
    }^{-1}
    \\[1mm]
    &= 
    (q + \alpha_2 p^{-1})^{-1} \,
    \brackets{
    q p - \alpha_4
    } \,
    \brackets{
    q p + \alpha_1 + \alpha_2
    }^{-1}
    \\[1mm]
    &= p \, (q p + \alpha_2)^{-1} \,
    (q p - \alpha_4) \, 
    (q p + \alpha_1 + \alpha_4)^{-1},
\end{align}
or, equivalently,
\begin{align}
    q \, \bar q
    &= t \, q p \, (q p + \alpha_2)^{-1} \,
    (q p - \alpha_4) \, 
    (q p + \alpha_1 + \alpha_4)^{-1}.
\end{align}
Let us define $g := q \, p$ and compute $\bar g = \bar q \, \bar p$. Note that
\begin{align}
    &&
    t^{-1} \bar q
    &= \brackets{
    p - (\alpha_2 + \alpha_4) (q + \alpha_2 p^{-1})^{-1}
    } \,
    \brackets{
    q p + \alpha_1 + \alpha_2
    }^{-1}
    \\
    \Leftrightarrow&&
    t^{-1} \bar q
    &= \brackets{
    p q + \alpha_1 + \alpha_2
    }^{-1} \, 
    \brackets{
    p - (\alpha_2 + \alpha_4) (q + \alpha_2 p^{-1})^{-1}
    }
    \\
    \Leftrightarrow&&
    p - (\alpha_2 + \alpha_4) (q + \alpha_2 p^{-1})^{-1}
    &= t^{-1} \brackets{
    p q + \alpha_1 + \alpha_2
    } \, \bar q.
\end{align}
Then,
\begin{align}
    \bar q \, \bar p
    &= - (1 + \alpha_2)
    - (p q + \alpha_1 + \alpha_2)
    + (\alpha_0 - 1) \, t \, (\bar q - t)^{-1}
    + (\alpha_3 - 1) \, (\bar q - 1)^{-1}.
\end{align}
Since $p q - q p = q^{-1} \, q p q - q p q \, q^{-1} = \LieBrackets{q^{-1}, q p \, q} = \LieBrackets{q^{-1}, g \, q}$ and $f := q$, we are done.
\end{proof}

A corresponding system of difference equations is given below:
\begin{gather}
    \tag*{d-P$_5$}
    \label{eq:dP5}
    \left\{
    \begin{array}{rcl}
         f_n \, f_{n + 1}
         &=& t \, g_n \, (g_n + \alpha_{2, n})^{-1} \,
         (g_n - \alpha_{4, n}) \, (g_n + \alpha_{1, n} + \alpha_{2, n})^{-1}
         ,
         \\[2mm]
         g_{n} + g_{n - 1}
         + \LieBrackets{f_{n - 1}^{-1}, g_{n - 1} \, f_{n - 1}}
         &=& (\alpha_{0, n} + \alpha_{3, n} + \alpha_{4, n})
         + \alpha_{3, n} \, (f_n - 1)^{-1}
         + \alpha_{0, n} \, t \, (f_n - t)^{-1}
         , 
    \end{array}
    \right.
    \\[2mm]
    \begin{aligned}
    \alpha_{0, n} 
    &= \alpha_{0} - n,
    &&&&&
    \alpha_{1, n} 
    &= \alpha_{1},
    &&&&&
    \alpha_{2, n} 
    &= \alpha_{2} + n,
    &&&&&
    \alpha_{3, n}
    &= \alpha_3 - n,
    &&&&&
    \alpha_{4, n}
    &= \alpha_4
    .
    \end{aligned}
\end{gather}

    \appendix
\section{\texorpdfstring{$d$}{d}-\Painleve equations}
\label{app:dP}
Here the variables $q$, $p$, $f$, $g \in \mathcal{R}$ and all constant parameters labeling by greek letters are from the field $\mathbb{C}$. The~element $t$ is central, except for the \ref{eq:dPE6}, \ref{eq:dPE6'}, and \ref{eq:dPE7} systems.

\subsection{\texorpdfstring{$d$-P$(D_4)$}{dPD4}}
Subsection \ref{sec:dPD4}. Affine Weyl group: $\widetilde{W} (D_4^{(1)}) = \angleBrackets{s_0, s_1, s_2, s_3, s_4; \pi_1, \pi_2, \pi_3}$.

\medskip
\textbullet \,\, $T~=~{s_2 s_4 s_1 s_2 s_0 s_3 \pi_2 \pi_3}$:

\begin{gather}
    \tag*{d-P$(D_4)$}
    \label{eq:dPD4}
    \begin{gathered}
    \begin{aligned}
        \begin{aligned}
        \bar \alpha_0 
        &= \alpha_0 - 1,
        &&&&
        \end{aligned} 
        \bar \alpha_2 
        = \alpha_2 + 1,&
        &&&&&
        &
        \bar \alpha_3 
        = \alpha_3 - 1,
        \\
        \begin{aligned}
        f \, \bar f
        = t \, g \, (g + \alpha_2)^{-1} \,
        (g - \alpha_4) \, (g + \alpha_1 + \alpha_2)^{-1},
        \\[1mm]
        \phantom{f}
        \end{aligned}&
        &&&&&
        &
        \begin{aligned}
        \bar g 
        &+ g
        + \LieBrackets{f^{-1}, \, g \, f}
        = (\alpha_0 + \alpha_3 + \alpha_4 - 2)
        \\[1mm]
        &+ (\alpha_3 - 1) \, (\bar f - 1)^{-1}
        + (\alpha_0 - 1) \, t \, (\bar f - t)^{-1}
        .
        \end{aligned}
    \end{aligned}
    \end{gathered}
\end{gather}

\subsection{\texorpdfstring{$d$-P$(D_5)$}{dPD5}}
Subsection \ref{sec:dPD5}. Affine Weyl group: $\widetilde{W} (A_3^{(1)}) = \angleBrackets{s_0, s_1, s_2, s_3; \pi}$.

\medskip
\textbullet \,\, Subsection \ref{sec:dPD5_1}. $T = (\pi_1 \pi_2)^2 s_0 s_2 s_1 s_3$: 

\begin{flalign}
    \tag*{d-P$(D_5)$}
    \label{eq:dPD5}
    &&
    \begin{gathered}
    \begin{aligned}
        \begin{aligned}
        \bar \alpha_0 
        &= \alpha_0 + 1,
        &&&&
        \end{aligned} 
        \bar \alpha_1 
        = \alpha_1- 1,&
        &&&&&
        &
        \begin{aligned}
        \bar \alpha_2 
        &= \alpha_2 + 1,
        &&&&&
        \bar \alpha_3 
        &= \alpha_3 - 1,
        \end{aligned}
        \\
        \bar q + q 
        = 1 - \alpha_2 p^{-1} - \alpha_0 (p + t)^{-1},&
        &&&&&
        &\bar p + p
        = - t + (\alpha_1 - 1) \bar q^{-1} + (\alpha_3 - 1) (\bar q - 1)^{-1}.
    \end{aligned}
    \end{gathered}
    &&
\end{flalign}

\medskip
\textbullet \,\, Subsection \ref{sec:dPD5_2}. $T' = s_3 s_0 s_1 \pi_2 \pi_1$:

\begin{gather}
    \tag*{d-P$(D_5)'$}
    \label{eq:dPD5'}
    \begin{gathered}
    \begin{aligned}
        \begin{aligned}
        \bar \alpha_2 
        &= \alpha_2 - 1,
        &&&&
        \end{aligned} 
        \bar \alpha_3 
        = \alpha_3 + 1,&
        &&&&&
        &
        \\
        \bar f + f + \LieBrackets{
        g^{-1}, \bar f \, g
        }   
        = - (\alpha_0 + \alpha_2) - t g - \alpha_2 (g - 1)^{-1},&
        &&&&&
        &g \, \bar g
        = - t^{-1} \bar f (\bar f + \alpha_0) (\bar f - \alpha_3)^{-1}.
    \end{aligned}
    \end{gathered}
\end{gather}

\subsection{\texorpdfstring{$d$-P$(D_6)$}{dPD6}}
Subsection \ref{sec:dPD6}. Affine Weyl group: $\widetilde{W} (2 A_1^{(1)}) = \angleBrackets{s_0, s_1, s_0', s_1'; \tilde \pi := \sigma \pi' \pi \sigma}$.

\medskip
\textbullet \,\, Subsection \ref{sec:dPD6_1}. $T = \pi' \sigma s_1 \sigma$: 
\begin{gather}
    \tag*{d-P$(D_6)$}
    \label{eq:dPD6}
    \begin{gathered}
    \begin{aligned}
        \bar \alpha_1
        = \alpha_1,&
        &&&&&
        &
        \bar \beta_1 
        = \beta_1 - 1,
        \\
        \bar f \, f
        = t + \bar \beta_1 \, t \, \bar g^{-1},&
        &&&&&
        &\bar g
        + g
        + [f \, \bar g, f^{-1}]
        &= \alpha_1 - \beta_1 + f + t f^{-1}.
    \end{aligned}
    \end{gathered}
\end{gather}

\medskip
\textbullet \,\, Subsection \ref{sec:dPD6_2}. $T' = (\sigma s_1)^2 \pi' \pi$:

\begin{gather}
    \tag*{d-P$(D_6)'$}
    \label{eq:dPD6'}
    \begin{gathered}
    \begin{aligned}
        \begin{aligned}
        \bar \alpha_0 
        &= \alpha_0 - 1,
        &&&&
        \end{aligned} 
        \bar \alpha_1 
        = \alpha_1 + 1,&
        &&&&&
        &
        \begin{aligned}
        \bar \beta_0 
        &= \beta_0 - 1,
        &&&&&
        \bar \beta_1 
        &= \beta_1 + 1,
        \end{aligned}
        \\
        \bar q + q
        = - \alpha_1 p^{-1} + \beta_1 (1 - p)^{-1},&
        &&&&&
        &\bar p
        + p
        = 1 - (\bar \alpha_1 + \beta_1) \bar q^{-1} - t \bar q^{-2}.
    \end{aligned}
    \end{gathered}
\end{gather}

\subsection{\texorpdfstring{$d$-P$(D_7)$}{dPD7}}
Subsection \ref{sec:dPD7}. Affine Weyl group and translation operator: $\widetilde{W} (A_1^{(1)}) = \angleBrackets{s_0, s_1; \pi}$, $T = \pi s_1$.

\medskip
\begin{gather}
    \tag*{d-P$(D_7)$}
    \label{eq:dPD7}
    \begin{gathered}
    \begin{aligned}
        \bar \alpha_0 
        = \alpha_0 - 1,&
        &&&&&
        &\bar \alpha_1 
        = \alpha_1 + 1,
        \\
        \bar f + f 
        + \LieBrackets{g \, \bar f, g^{-1}}
        = - \alpha_1 - t g^{-1},&
        &&&&&
        &\bar g \, g
        = t \bar f.
    \end{aligned}
    \end{gathered} 
\end{gather}

\subsection{\texorpdfstring{$d$-P$(E_6)$}{dPE6}}
Subsection \ref{sec:dPE6}. Affine Weyl group: $\widetilde{W} (A_2^{(1)}) = \angleBrackets{s_0, s_1, s_2; \pi}$.

\medskip
\textbullet \,\, Subsection \ref{sec:dPE6_1}. $T = s_1 \pi s_2$: 

\begin{gather}
    \tag*{d-P$(E_6)$}
    \label{eq:dPE6}
    \begin{gathered}
    \begin{aligned}
        \bar \alpha_1 
        = \alpha_1 + 1,&
        &&&&&
        &\bar \alpha_2 
        = \alpha_2 - 1,
        \\
        \bar q
        + q 
        = - t + \bar p - (\alpha_2 - 1) {\bar{p}}^{-1},&
        &&&&&
        &\bar p
        + p 
        = t + q + \alpha_1 q^{-1}.
    \end{aligned}
    \end{gathered}
\end{gather}

\medskip
\textbullet \,\, Subsection \ref{sec:dPE6_2}. $T' = s_2 s_1 \pi$:

\begin{gather}
    \tag*{d-P$(E_6)'$}
    \label{eq:dPE6'}
    \begin{gathered}
    \begin{aligned}
        \bar \alpha_0 
        = \alpha_0 - 1,&
        &&&&&
        &\bar \alpha_2 
        = \alpha_2 + 1,
        \\
        \bar f \, f
        = - (\bar g - \alpha_1) (\bar g + \alpha_2)^{-1} \bar g,&
        &&&&&
        &\bar g
        + g
        + \LieBrackets{f \, \bar g, f^{-1}}
        = \alpha_1
        + f t + f^2.
    \end{aligned}
    \end{gathered}
\end{gather}

\subsection{\texorpdfstring{$d$-P$(E_7)$}{dPE7}}
Subsection \ref{sec:dPE7}. Affine Weyl group and translation operator: $\widetilde{W} (A_1^{(1)}) = \angleBrackets{s_0, s_1; \pi}$, $T = \pi s_1$.

\medskip
\begin{gather}
    \tag*{d-P$(E_7)$}
    \label{eq:dPE7}
    \begin{gathered}
    \begin{aligned}
        \bar \alpha_0 
        = \alpha_0 - 1,&
        &&&&&
        &\bar \alpha_1 
        = \alpha_1 + 1,
        \\
        \bar q
        + q 
        = - \alpha_1 p^{-1},&
        &&&&&
        &\bar p
        + p 
        = t + 2 {\bar{q}}^2.
    \end{aligned}
    \end{gathered}
\end{gather}
\section{Degeneration data}
\label{app:deg_data}

Here capital letters correspond to the lower equation, while $\varepsilon$ is a small parameter. The scheme is as below
\begin{figure}[H]
    \centering
    \scalebox{1.}{\tikzset{every picture/.style={line width=0.75pt}} 

\begin{tikzpicture}[x=0.75pt,y=0.75pt,yscale=-1,xscale=1]

\draw    (290.63,69.79) -- (320.13,59.01) ;
\draw [shift={(322.01,58.33)}, rotate = 159.93] [color={rgb, 255:red, 0; green, 0; blue, 0 }  ][line width=0.75]    (10.93,-3.29) .. controls (6.95,-1.4) and (3.31,-0.3) .. (0,0) .. controls (3.31,0.3) and (6.95,1.4) .. (10.93,3.29)   ;
\draw    (387.63,117.79) -- (416.34,104.4) ;
\draw [shift={(418.15,103.56)}, rotate = 155] [color={rgb, 255:red, 0; green, 0; blue, 0 }  ][line width=0.75]    (10.93,-3.29) .. controls (6.95,-1.4) and (3.31,-0.3) .. (0,0) .. controls (3.31,0.3) and (6.95,1.4) .. (10.93,3.29)   ;
\draw    (290.63,100.79) -- (320.05,115.05) ;
\draw [shift={(321.85,115.92)}, rotate = 205.86] [color={rgb, 255:red, 0; green, 0; blue, 0 }  ][line width=0.75]    (10.93,-3.29) .. controls (6.95,-1.4) and (3.31,-0.3) .. (0,0) .. controls (3.31,0.3) and (6.95,1.4) .. (10.93,3.29)   ;
\draw    (388.63,57.79) -- (418.05,72.05) ;
\draw [shift={(419.85,72.92)}, rotate = 205.86] [color={rgb, 255:red, 0; green, 0; blue, 0 }  ][line width=0.75]    (10.93,-3.29) .. controls (6.95,-1.4) and (3.31,-0.3) .. (0,0) .. controls (3.31,0.3) and (6.95,1.4) .. (10.93,3.29)   ;
\draw    (389.78,50.63) -- (417.01,50.63) ;
\draw [shift={(419.01,50.63)}, rotate = 180] [color={rgb, 255:red, 0; green, 0; blue, 0 }  ][line width=0.75]    (10.93,-3.29) .. controls (6.95,-1.4) and (3.31,-0.3) .. (0,0) .. controls (3.31,0.3) and (6.95,1.4) .. (10.93,3.29)   ;
\draw    (487.63,57.99) -- (517.05,72.25) ;
\draw [shift={(518.85,73.12)}, rotate = 205.86] [color={rgb, 255:red, 0; green, 0; blue, 0 }  ][line width=0.75]    (10.93,-3.29) .. controls (6.95,-1.4) and (3.31,-0.3) .. (0,0) .. controls (3.31,0.3) and (6.95,1.4) .. (10.93,3.29)   ;
\draw    (488.78,50.83) -- (516.01,50.83) ;
\draw [shift={(518.01,50.83)}, rotate = 180] [color={rgb, 255:red, 0; green, 0; blue, 0 }  ][line width=0.75]    (10.93,-3.29) .. controls (6.95,-1.4) and (3.31,-0.3) .. (0,0) .. controls (3.31,0.3) and (6.95,1.4) .. (10.93,3.29)   ;
\draw    (488.8,84.58) -- (516.03,84.58) ;
\draw [shift={(518.03,84.58)}, rotate = 180] [color={rgb, 255:red, 0; green, 0; blue, 0 }  ][line width=0.75]    (10.93,-3.29) .. controls (6.95,-1.4) and (3.31,-0.3) .. (0,0) .. controls (3.31,0.3) and (6.95,1.4) .. (10.93,3.29)   ;

\draw (329.13,43.37) node [anchor=north west][inner sep=0.75pt]    {\ref{eq:dPD5'}};
\draw (331.1,113.73) node [anchor=north west][inner sep=0.75pt]    {\ref{eq:dPD5}};
\draw (428.1,78.13) node [anchor=north west][inner sep=0.75pt]    {\ref{eq:dPE6}};
\draw (428.1,43.4) node [anchor=north west][inner sep=0.75pt]    {\ref{eq:dPD6}};
\draw (236.35,78.12) node [anchor=north west][inner sep=0.75pt]    {\ref{eq:dPD4}};
\draw (527.1,78.33) node [anchor=north west][inner sep=0.75pt]    {\ref{eq:dPE7}};
\draw (527.1,43.6) node [anchor=north west][inner sep=0.75pt]    {\ref{eq:dPD7}};

\end{tikzpicture}}
\end{figure}

\subsection{\ref{eq:dPD4} $\to \{ \text{\ref{eq:dPD5}}, \, \text{\ref{eq:dPD5'}} \}$}
\phantom{}
\begin{itemize}
\item[\textbullet]
\ref{eq:dPD4} $\to$ \ref{eq:dPD5}:
(see eq. (8.127) in \cite{kajiwara2017geometric})
\begin{gather}
    \begin{aligned}
    f
    &= 1 + \varepsilon T \, Q,
    &&&&&
    g
    &= \varepsilon^{-1} \, T^{-1} \, P,
    &&&&&
    t
    &= 1 + \varepsilon \, T,
    \end{aligned}
    \\[1mm]
    \begin{aligned}
    \alpha_0
    &= \mathrm{A}_3,
    &&&&&
    \alpha_1
    &= - \varepsilon^{-1} + \mathrm{A}_2,
    &&&&&
    \alpha_2
    &= \varepsilon^{-1},
    &&&&&
    \alpha_3 
    &= \mathrm{A}_1,
    &&&&&
    \alpha_4
    &= - \varepsilon^{-1} + \mathrm{A}_0.
    \end{aligned}
\end{gather}

\item[\textbullet]
\ref{eq:dPD4} $\to$ \ref{eq:dPD5'}:
\begin{gather}
    \begin{aligned}
    f
    &= \varepsilon \, T \, G,
    &&&&&
    g
    &= \varepsilon \, F,
    &&&&&
    t
    &= \varepsilon \, T,
    \end{aligned}
    \\[1mm]
    \begin{aligned}
    \alpha_0
    &= 1 - \varepsilon \, \mathrm{A}_2,
    &&&&&
    \alpha_1
    &= 1 + \varepsilon,
    &&&&&
    \alpha_2
    &= - 1 - \varepsilon \, (1 + \mathrm{A}_3),
    &&&&&
    \alpha_3 
    &= 2,
    &&&&&
    \alpha_4
    &= - \varepsilon \, \mathrm{A}_0.
    \end{aligned}
\end{gather}
\end{itemize}

\subsection{$\{ \text{\ref{eq:dPD5}}, \, \text{\ref{eq:dPD5'}} \} \to \text{\ref{eq:dPE6}}$}
\phantom{}
\begin{itemize}
\item[\textbullet]
\ref{eq:dPD5} $\to$ \ref{eq:dPE6}: 
(see eq. (8.128) in \cite{kajiwara2017geometric})
\begin{gather}
    \begin{aligned}
    q
    &= \varepsilon \, Q,
    &&&&&
    p
    &= \varepsilon^{-2} + \varepsilon^{-1} (P - T)
    ,
    &&&&&
    t
    &= \varepsilon^{-1} T - \varepsilon^{-2}
    ,
    \end{aligned}
    \\[1mm]
    \begin{aligned}
    \alpha_0
    &= \mathrm{A}_1,
    &&&&&
    \alpha_1
    &= \mathrm{A}_2,
    &&&&&
    \alpha_2
    &= \varepsilon^{-2},
    &&&&&
    \alpha_3 
    &= - \varepsilon^{-2}.
    \end{aligned}
\end{gather}

\item[\textbullet]
\ref{eq:dPD5'} $\to$ \ref{eq:dPE6}:
(cf. \cite{grammaticos1998degeneration})
\begin{gather}
    \begin{aligned}
    f
    &= 1 + \varepsilon \, Q,
    &&&&&
    g
    &= 1 + \varepsilon\, P
    ,
    &&&&&
    t
    &= - 1 + \varepsilon \, T
    ,
    \end{aligned}
    \\[1mm]
    \begin{aligned}
    \alpha_0
    &= - 1 + \varepsilon^2 \, \mathrm{A}_1,
    &&&&&
    \alpha_1
    &= \mathrm{A}_1,
    &&&&&
    \alpha_2
    &= \varepsilon^{2} \, \mathrm{A}_2,
    &&&&&
    \alpha_3 
    &= 1.
    \end{aligned}
\end{gather}

\end{itemize}

\subsection{$\text{\ref{eq:dPD5'}} \to \text{\ref{eq:dPD6}}$}
\phantom{}
\begin{itemize}
\item[\textbullet]
\ref{eq:dPD5'} $\to$ \ref{eq:dPD6}: 
(cf. \cite{grammaticos1998degeneration})
\begin{gather}
    \begin{aligned}
    f
    &= - \mathrm{B}_1^2 \, \varepsilon^{-1} \, F,
    &&&&&
    g
    &= 1 + \mathrm{B}_1^{-1} \, G
    ,
    &&&&&
    t
    &= \mathrm{B}_1^{-3} \, \varepsilon
    ,
    \end{aligned}
    \\[1mm]
    \begin{aligned}
    \alpha_0
    &= - 2 - \mathrm{B}_1^{-1} (\mathrm{A}_1 - \mathrm{B}_1)
    - \mathrm{B}_1 \, \varepsilon^{-1} \, T
    ,
    &&&&&
    \alpha_2
    &= \varepsilon^{-1} \, \mathrm{B}_1 \, T,
    &&&&&
    \alpha_3 
    &= 1.
    \end{aligned}
\end{gather}
\end{itemize}

\subsection{$\text{\ref{eq:dPD6}} \to \{ \text{\ref{eq:dPD7}}, \text{\ref{eq:dPE7}} \}$}
\phantom{}
\begin{itemize}

\item[\textbullet]
\ref{eq:dPD6} $\to$ \ref{eq:dPD7}:
\begin{align}
    f
    &= - \varepsilon^2 \, G,
    &
    g
    &= 1 + \varepsilon \, F
    ,
    &
    t
    &= \varepsilon^3 \, T
    ,
    &
    \alpha_1
    &= - \varepsilon \, \rm{A}_1,
    &
    \beta_1
    &= - 2.
\end{align}

\item[\textbullet]
\ref{eq:dPD6} $\to$ \ref{eq:dPE7}:
(cf. \cite{grammaticos1998degeneration})
\begin{align}
    f
    &= 1 + \sqrt{2} \, \varepsilon \, Q,
    &
    g
    &= \varepsilon^2 \, P
    ,
    &
    t
    &= 1
    ,
    &
    \alpha_1
    &= - 3 + \varepsilon^2 \, T,
    &
    \beta_1
    &= -1 - \sqrt{2} \varepsilon^{3} \, \rm{A}_1.
\end{align}
\end{itemize}

\subsection{$\text{\ref{eq:dPE6}} \to$ \ref{eq:dPE7}}
\phantom{}
\begin{itemize}
\item[\textbullet]
\ref{eq:dPE6} to \ref{eq:dPE7}: (cf. \cite{grammaticos1998degeneration})
\begin{align}
    q
    &= 1 + \sqrt{2} \varepsilon \, Q,
    &
    p
    &= \varepsilon^2 \, P,
    &
    t
    &= - 2 + \varepsilon^2 \, T,
    &
    \alpha_1
    &= 1,
    &
    \alpha_2
    &= 1 + \sqrt{2} \varepsilon^3 \, \mathrm{A}_1
    .
\end{align}
\end{itemize}

    \bibliographystyle{alpha}
    \bibliography{bib}
\end{document}